\documentclass[envcountsect,runningheads]{llncs}

\usepackage[final]{graphics}
\usepackage{color,url}
\usepackage[english]{babel}
\usepackage[utf8]{inputenc}
\usepackage{amsmath,amssymb,mathtools}
\usepackage[final]{listings}
\usepackage{wrapfig}
\usepackage{tikz}
\usetikzlibrary{arrows}

\title{Validating Mathematical Structures}
\author{Kazuhiko Sakaguchi}
\institute{%
University of Tsukuba, Japan\\
\texttt{sakaguchi@logic.cs.tsukuba.ac.jp}
}

\definecolor{dkblue}{rgb}{0,0.1,0.5}
\definecolor{lightblue}{rgb}{0,0.5,0.5}
\definecolor{dkgreen}{rgb}{0,0.4,0}
\definecolor{dk2green}{rgb}{0.4,0,0}
\definecolor{dkviolet}{rgb}{0.6,0,0.8}
\definecolor{dkpink}{rgb}{0.2,0,0.6}
\definecolor{redbrown}{rgb}{0.65, 0.16, 0.16}
\definecolor{commentcolor}{rgb}{0.2, 0.5, 0.5}
\lstdefinelanguage{SSR} {
mathescape=true,
escapeinside={[*}{*]},
texcl=false,
morekeywords=[1]{
From, Section, Module, End, Require, Import, Export, Defensive, Function,
Variable, Variables, Parameter, Parameters, Axiom, Hypothesis, Hypotheses,
Notation, Local, Global, Tactic, Reserved, Scope, Open, Close, Bind, Delimit,
Definition, Let, Ltac, Fixpoint, CoFixpoint, Add, Morphism, Relation,
Implicit, Arguments, Unset, Contextual, Strict, Prenex, Implicits,
Inductive, CoInductive, Record, Variant, Structure, Canonical, Coercion,
Theorem, Lemma, Corollary, Proposition, Fact, Remark, Example,
Proof, Goal, Save, Qed, Defined, Admitted, Abort, Restart,
Opaque, Transparent,
Hint, Resolve, Rewrite, View,
Search, Show, Print, Check, About, Locate, Printing, All, Graph, Projections,
Extraction, Extract, Inlined, Constant},
moredelim=[is][\color{dkviolet}]{|>}{<|},
moredelim=[is][\color{red}]{|*}{*|},
morekeywords=[2]{forall, exists, exists2, fun, fix, cofix, struct,
      match, with, end, as, in, return, let, if, is, then, else,
      for, of, nosimpl, mlet},
morekeywords=[3]{SProp, Prop, Set, Type},
morekeywords=[4]{
         pose, set, move, case, elim, apply, clear,
            hnf, intro, intros, generalize, rename, pattern, after,
	    destruct, induction, using, refine, inversion, injection,
         rewrite, congr, unlock, compute, ring, field,
            replace, fold, unfold, change, cutrewrite, simpl,
         have, gen, generally, suff, wlog, suffices, without, loss, nat_norm,
            assert, cut, trivial, revert, bool_congr, nat_congr, abstract,
	 symmetry, transitivity, auto, split, left, right, autorewrite},
morekeywords=[5]{
         by, done, exact, reflexivity, tauto, romega, omega,
         assumption, solve, contradiction, discriminate, congruence},
morekeywords=[6]{do, last, first, try, idtac, repeat},
 literate=*
	{->}{{$\rightarrow\,$}}2
	{<-}{{$\leftarrow\,$}}2
 	{>->}{{>->}}3
	{|-}{{$\vdash$}}2
 	{<=}{{$\leq$}}1
 	{>=}{{$\geq$}}1
 	{<>}{{$\neq$}}1
 	{/\\}{{$\wedge$}}2
 	{\\/}{{$\vee$}}2
 	{<->}{{$\leftrightarrow\;$}}3
 	{<=>}{{$\Leftrightarrow\;$}}3
 	{:nat}{{$~\in\mathbb{N}$}}3
	{fforall\ }{{$\forall_f\,$}}1
	{forall\ }{{$\forall\,$}}1
 	{negb}{{$\neg$}}1
 	{spp}{{:*:\,}}1
 	{~~}{{$\neg$}}1
 	{\\in}{{$\in\;$}}1
 	{/\\}{$\land\,$}1
 	{:*:}{{$*$}}2
	{=>}{{$\Rightarrow$}}2
 	{!=}{{$\neq$}\,}2
 	{^-1}{{$^{-1}$}}1
 	{elt'}{elt'}1
	{isn't }{{{\ttfamily\color{dkgreen} isn't }}}1,
showstringspaces=false,
morestring=[b]",
morestring=[d]",
morecomment=[s]{(*}{*)},
tabsize=3,
extendedchars=true,
sensitive=true,
identifierstyle={\ttfamily\color{black}},
keywordstyle=[1]{\ttfamily\color{dkviolet}},
keywordstyle=[2]{\ttfamily\color{dkgreen}},
keywordstyle=[3]{\ttfamily\color{redbrown}},
keywordstyle=[4]{\ttfamily\color{dkblue}},
keywordstyle=[5]{\ttfamily\color{red}},
keywordstyle=[6]{\ttfamily\color{dkpink}},
stringstyle=\ttfamily,
commentstyle={\ttfamily\color{commentcolor}},
}

\lstdefinestyle{plain}{
  basicstyle=\ttfamily,
  keywordstyle=,
  identifierstyle=,
  commentstyle=,
  stringstyle=,
  emphstyle=,
  backgroundcolor=,
  language=,
  frame=tlbr,
  framesep=0pt,
  rulecolor=\color{white},
  numbers=none,
  numberstyle=\tiny,
  xleftmargin=0pt,
  xrightmargin=0pt,
  basewidth=0.459em,
  lineskip=-.56ex,
  aboveskip=4pt,
  belowskip=-2pt,
  keepspaces,
}

\lstnewenvironment{code}[1][]{
  \lstset{
  basicstyle=\ttfamily\small,
  breaklines=false,
  #1}
}{}

\lstnewenvironment{coqcode}[1][]{
  \lstset{
  basicstyle=\ttfamily\small,
  breaklines=false,
  language=SSR,
  #1}
}{}

\newcommand\coqinline[1][]%
{\lstinline[language=SSR, breaklines=true, breakatwhitespace=true, #1]}

\tikzset{
every node/.style={inner sep=2},
every path/.style={thick},
every picture/.style={font issue=\footnotesize},
font issue/.style={execute at begin picture={#1\selectfont}}
}

\newcommand\conv{\equiv}
\newcommand\unify{\ensuremath{\mathrel{\hat{\conv}}}}
\newcommand\leadsfrom{\mathrel{\vcenter{\hbox{\rotatebox{180}{$\leadsto$}}}}}
\newcommand\evar[1]{\ensuremath{?_{\texttt{#1}}}}

\newcommand*{\qedhere}{\strut~\hfill$\qed$\endproof}

\usepackage{algorithm2e}

\SetCommentSty{mycommfont}

\newcommand{\Coq}{{\sffamily Coq}}
\newcommand{\Gallina}{{\sffamily Gallina}}
\newcommand{\Ltac}{{\sffamily Ltac}}
\newcommand{\MC}{{\sffamily MathComp}}
\newcommand{\Analysis}{{\sffamily MathComp Analysis}}
\newcommand{\Lean}{{\sffamily Lean}}
\newcommand{\Matita}{{\sffamily Matita}}
\newcommand{\OCaml}{{\sffamily OCaml}}

\newcommand{\Coquelicot}{{\sffamily Coquelicot}}
\newcommand{\kp}[1]{\textcolor{red}{KP -- #1}}

\renewcommand{\kp}[1]{}

\begin{document}

\maketitle

\begin{abstract}
\kp{The abstract (and intro) flow is not at all straightforward to follow. Recall that the abstract is read at least one order of magnitude more than other parts of the paper. Hence, here is a proposal for better flow:
\begin{enumerate}
\item Mechanisms for sharing notations, definitions, and results across an inheritance hierarchy of mathematical structures, e.g., a hierarchy of monoids, groups, and rings, are important for productivity when formalizing mathematics in proof assistants.
\item The packed classes method is a generic design pattern to define and combine mathematical structures in a dependent type theory with records.
\item When combined with mechanisms for implicit coercions and canonical structures, e.g., those in the \Coq{} proof assistant, packed classes enable automatic structure inference and subtyping in hierarchies, e.g., that a ring can be used in place of a monoid.
\item However, large hierarchies based on packed classes, such as those in the \MC{} library for \Coq, are challenging to design and maintain.
\item We identify two hierarchy invariants that ensure modularity of reasoning and predictability of inference with packed classes, and propose algorithms to check those invariants.
\item We implemented our algorithms in tools for \Coq, and applied them on the \MC{} library.
\item The results show that our tools significantly improve the development process for \MC.
\end{enumerate}
}

Sharing of notations and theories across an inheritance hierarchy of mathematical structures, e.g., groups and rings, is important for productivity when formalizing mathematics in proof assistants.
The packed classes methodology is a generic design pattern to define and combine mathematical structures in a dependent type theory with records.
When combined with mechanisms for implicit coercions and unification hints, packed classes enable automated structure inference and subtyping in hierarchies, e.g., that a ring can be used in place of a group.
However, large hierarchies based on packed classes are challenging to implement and maintain.
We identify two hierarchy invariants that ensure modularity of reasoning and predictability of inference with packed classes, and propose algorithms to check these invariants.
We implement our algorithms as tools for the \Coq{} proof assistant, and show that they significantly improve the development process of {\sffamily Mathematical Components}, a library for formalized mathematics.

\end{abstract}

\kp{General problems across the text:
\begin{itemize}
\item overuse of the format ``$\cdots$; thus, $\cdots$'' --- better to split into several sentences sometimes: ``$\cdots$. Thus, $\cdots$'' --- also use other words more like ``consequently,'', ``hence,''
\item light grey color for code is not readable, especially when the paper is printed --- also consider having very distinct different colors for emphasis of code and code comments
\item instead of using the quite ugly \cite[Sect.~``Implicit Coercions'']{coqrefman}, make a different BiBTeX entry for different key chapters of the refman (if you're going to link to URLs anyway, use the specific HTML page for the chapter)
\item there needs to be a clearer distinction (or more consistent terminology) between ``provers'', ``proof assistants'', and ``systems [like Coq]'' -- it looks like sometimes they are used as synonyms, and sometimes not
\end{itemize}
}

\section{Introduction}
\label{sec:introduction}

\begin{figure*}[t]
 \centering
 \includegraphics[width=\textwidth,pagebox=cropbox]{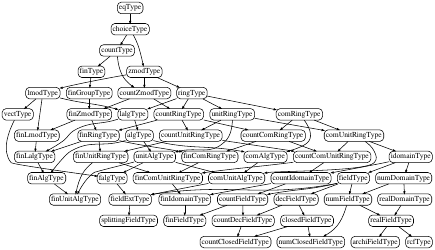}
 \caption{The hierarchy of structures in the \MC{} library 1.10.0}
 \label{fig:hierarchy}
\end{figure*}

Mathematical structures are a key ingredient of modern formalized mathematics in proof assistants, e.g.,~\cite{Geuvers:2002,Spitters:2011,Holzl:2013,forgetful-inference}\cite[Chap.~2 and Chap.~4]{Cohen:phd}\cite[Sect.~3]{Gonthier:2013}\cite[Chap.~5]{Rouhling:2019}\cite[Sect.~4]{mathlib:2020}.
Since mathematical structures have an inheritance/subtyping hierarchy such that ``a ring is a group and a group is a monoid'', it is usual practice in mathematics to reuse notations and theories of superclasses implicitly to reason about a subclass.
Similarly, the sharing of notations and theories across the hierarchy is important for productivity when formalizing mathematics.

The packed classes methodology~\cite{Garillot:2009,Garillot:2011} is a generic design pattern to define and combine mathematical structures in a dependent type theory with records. Hierarchies using packed classes support multiple inheritance, and maximal sharing notations and theories.
When combined with mechanisms for implicit coercions~\cite{Saibi:1997,Saibi:1999} and for extending unification procedure, such as the canonical structures~\cite{Saibi:1999,Mahboubi:2013} of the \Coq{} proof assistant~\cite{coqrefman}, and the unification hints~\cite{Asperti:2009} of the \Lean{} theorem prover~\cite{leanprover,deMoura:2015} and the \Matita{} interactive theorem prover~\cite{Asperti:2011}, packed classes enable subtyping and automated inference of structures in hierarchies.
Compared to approaches based on type classes~\cite{Sozeau:2008,Haftmann:2007}, packed classes are more robust, and their inference approach is efficient and \pagebreak predictable~\cite{forgetful-inference}.
The success of the packed classes methodology in formalized mathematics can be seen in the {\sffamily Mathematical Components} library~\cite{mathcomp-github} (hereafter \MC), the \Coquelicot{} library~\cite{Boldo:2015}
, and especially the formal proof of the Odd Order Theorem~\cite{Gonthier:2013}.
 It has also been successfully applied for program verification tasks, e.g., a hierarchy of monadic effects~\cite{Affeldt:2019} and a hierarchy of partial commutative monoids~\cite{fcsl-pcm} for Fine-grained Concurrent Separation Logic~\cite{Sergey:PLDI:2015}.

In spite of its success, the packed classes methodology is hard to master for library designers and requires a substantial amount of work to maintain as libraries evolve.
For instance, the strict application of packed classes requires defining quadratically many implicit coercions and unification hints in the number of structures.
To give some figures, the \MC{} library 1.10.0 uses this methodology ubiquitously to define the 51 mathematical structures depicted in Fig.~\ref{fig:hierarchy}, and declares 554 implicit coercions and 746 unification hints to implement their inheritance.
Moreover, defining new intermediate structures between existing ones requires fixing their subclasses and their inheritance accordingly; thus, it can be a challenging task.

In this paper, we indentify two hierarchy invariants concerning implicit coercions and unification hints in packed classes, and propose algorithms to check these invariants.
We implement our algorithms as tools for the \Coq{} system, evaluate our tools on a large-scale development, the \MC{} library 1.7.0, and then successfully detect and fix several inheritance bugs with the help of our tools.
The invariant concerning implicit coercions ensures the modularity of reasoning with packed classes and is also useful in other approaches, such as type classes and telescopes~\cite[Sect.~2.3]{Mahboubi:2013}, in a dependent type theory.
This invariant was proposed before as a \emph{coherence} of inheritance graphs~\cite{Barthe:1995}. \pagebreak
The invariant concerning unification hints, that we call \emph{well-formedness}, ensures the predictability of structure inference.
Our tool not only checks well-formedness, but also generates an exhaustive set of assertions for structure inference, and these assertions can be tested inside \Coq.
We state the predictability of inference as a metatheorem on a simplified model of hierarchies, that we formally prove in \Coq.

The paper is organized as follows:
Section \ref{sec:packed-classes} reviews the packed classes methodology using a running example.
Section \ref{sec:coherence} studies the implicit coercion mechanism of \Coq{}, and then presents the new coherence checking algorithm and its implementation.
Section \ref{sec:automated-structure-inference} reviews the use of canonical structures for structure inference in packed classes, and introduces the notion of well-formedness.
Section \ref{sec:formal-hierarchy} defines a simplified model of hierarchies and structure inference, and shows the metatheorem that states the predictability of structure inference.
Section \ref{sec:validating-canonical-projections} presents the well-formedness checking algorithm and its implementation.
Section \ref{sec:evaluation} evaluates our checking tools on the \MC{} library 1.7.0.
Section \ref{sec:conclusion} discusses related work and concludes the paper.
Our running example for Sect.~\ref{sec:packed-classes}, Sect.~\ref{sec:automated-structure-inference}, and Sect.~\ref{sec:validating-canonical-projections}, the formalization for Sect.~\ref{sec:formal-hierarchy}, and the evaluation script for Sect.~\ref{sec:evaluation} are available at~\cite{supplementary}.

\section{Packed classes}
\label{sec:packed-classes}

This section reviews the packed classes methodology~\cite{Garillot:2009,Garillot:2011} through an example, but elides canonical structures.
Our example is a minimal hierarchy with multiple inheritance, consisting of the following four algebraic structures (Fig.~\ref{fig:example-hierarchy}):
\\\vspace{-5ex}
\begin{wrapfigure}[19]{r}{48mm}
 \centering
 \begin{tikzpicture}[structure/.style={draw, rounded corners=1mm}, x=13mm, y=7mm]
  \node[structure] (Type)     at (0,  3) {\coqinline{Type}};
  \node[structure] (Monoid)   at (0,  2) {\coqinline{Monoid.type}};
  \node[structure] (Semiring) at (-1, 1) {\coqinline{Semiring.type}};
  \node[structure] (Group)    at (+1, 1) {\coqinline{Group.type}};
  \node[structure] (Ring)     at (0,  0) {\coqinline{Ring.type}};
  \draw[->, densely dotted] (Type)     -- (Monoid);
  \draw[->]                 (Monoid)   -- (Semiring);
  \draw[->]                 (Monoid)   -- (Group);
  \draw[->]                 (Semiring) -- (Ring);
  \draw[->]                 (Group)    -- (Ring);
 \end{tikzpicture}
 \caption{Hierarchy diagram for monoids, semirings, groups, and rings, where an arrow from $\texttt{X.type}$ to $\texttt{Y.type}$ means that \coqinline{Y} directly inherits from \coqinline{X}.
 The monoid structure is the superclass of all other structures.
 Semirings and groups directly inherit from monoids.
 Rings directly inherit from semirings and groups, and indirectly inherit from monoids.}
 \label{fig:example-hierarchy}
\end{wrapfigure}
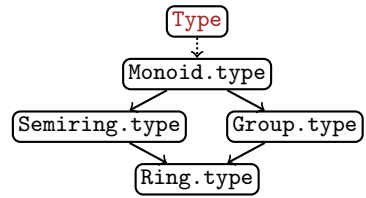
\begin{description}
 \item[Additive monoids $(A, +, 0)$:]
	    Monoids have an associative binary operation $+$ on the set $A$ and an identity element $0 \in A$.
 \item[Semirings $(A, +, 0, \times, 1)$:]
	    Semirings have the monoid axioms, commutativity of addition, multiplication, and an element $1 \in A$.
	    Multiplication $\times$ is an associative binary operation on $A$ that is left and right distributive over addition.
	    $0$ and $1$ are absorbing and identity elements with respect to multiplication, respectively.
 \item[Additive groups $(A, +, 0, -)$:]
	    Groups have the monoid axioms and a unary operation $-$ on $A$.
	    $- x$ is the additive inverse of $x$ for any $x \in A$.
 \item[Rings $(A, +, 0, -, \times, 1)$:]
	    Rings have all the semiring and group axioms, but no additional axioms.
\end{description}

We start by defining the base class, namely, the \coqinline{Monoid} structure.
\begin{coqcode}[numbers=left]
Module Monoid.

Record mixin_of (A : Type) := Mixin { [*\label{line:monoid-mixin_of}*]
  zero : A;
  add : A -> A -> A;
  addA : associative add;               (* `add' is associative.              *)
  add0x : left_id zero add;             (* `zero' is the left and right       *)
  addx0 : right_id zero add;            (*   identity element w.r.t. `add'.   *)
}.

Record class_of (A : Type) := Class { mixin : mixin_of A }. [*\label{line:monoid-class_of}*]

Structure type := Pack { sort : Type; class : class_of sort }. [*\label{line:monoid-type}*]

End Monoid.
\end{coqcode}
The above definitions are enclosed by the \coqinline{Monoid} module, which forces users to write qualified names such as \coqinline{Monoid.type}.
Thus, we can reuse the same name of record types (\coqinline{mixin_of}, \coqinline{class_of}, and \coqinline{type}), their constructors (\coqinline{Mixin}, \coqinline{Class}, and \coqinline{Pack}), and constants (e.g., \coqinline{sort} and \coqinline{class}) for other structures to indicate their roles.
Structures are written as records that have three different roles: \emph{mixins}, \emph{classes}, and \emph{structures}.
The mixin record (line \ref{line:monoid-mixin_of}) gathers operators and axioms newly introduced by the \coqinline{Monoid} structure. Since monoids do not inherit any other structure in Fig.~\ref{fig:example-hierarchy}, those are all the monoid operators, namely $0$ and $+$, and their axioms (see Appendix \ref{sec:coq-references} for details).
The class record (line \ref{line:monoid-class_of}) assembles all the mixins of the superclasses of the \coqinline{Monoid} structure (including itself), which is the singleton record consisting of the monoid mixin.
The structure record (line \ref{line:monoid-type}) is the actual interface of the structure that bundles a carrier of type \coqinline{Type} and its class instance.
\coqinline{Record} and \coqinline{Structure} are synonyms in \Coq, but we reserve the latter for actual interfaces of structures.
In a hierarchy of algebraic structures, a carrier set \coqinline{A} has type \coqinline{Type}; hence, for each structure, the first field of the \coqinline{type} record should have type \coqinline{Type}, and the \coqinline{class_of} record should be parameterized by that carrier.
In general, it can be other types, e.g., \coqinline{Type -> Type} in the hierarchy of functors and monads~\cite{Affeldt:2019}, but should be fixed in each hierarchy of structures.

Mixin and monoid records are internal definitions of mathematical structures; in contrast, the structure record is a part of the interface of the monoid structure when reasoning about monoids.
For this reason, we lift the projections for \coqinline{Monoid.mixin_of} to definitions and lemmas for \coqinline{Monoid.type} as follows.

\begin{coqcode}[alsoletter={.}, morekeywords={[7]Monoid.sort}, keywordstyle={[7]\ttfamily\color{commentcolor}}]
Definition zero {A : Monoid.type} : Monoid.sort A :=
  Monoid.zero _ (Monoid.mixin _ (Monoid.class A)).
Definition add {A : Monoid.type} :
  Monoid.sort A -> Monoid.sort A -> Monoid.sort A :=
  Monoid.add _ (Monoid.mixin _ (Monoid.class A)).
Lemma addA {A : Monoid.type} : associative (@add A).
Lemma add0x {A : Monoid.type} : left_id (@zero A) (@add A).
Lemma addx0 {A : Monoid.type} : right_id (@zero A) (@add A).
\end{coqcode}
The curly brackets enclosing \coqinline{A} mark it as an implicit argument; in contrast, \coqinline{@} is the explicit application symbol that deactivates the hiding of implicit arguments.

Since a monoid instance \coqinline{A : Monoid.type} can be seen as a type equipped with monoid axioms, it is natural to declare \coqinline{Monoid.sort} as an implicit coercion.
The types of \coqinline{zero} can be written and shown as \coqinline{forall A : Monoid.type, A} rather than \coqinline[alsoletter={.}, morekeywords={[7]Monoid.sort}, keywordstyle={[7]\ttfamily\color{commentcolor}}]{forall A : Monoid.type, Monoid.sort A} thanks to this implicit coercion.

\begin{coqcode}
Coercion Monoid.sort : Monoid.type >-> Sortclass.
\end{coqcode}

Next, we define the \coqinline{Semiring} structure.
Since semirings inherit from monoids and the semiring axioms interact with the monoid operators, e.g., distributivity of multiplication over addition, the semiring mixin should take \coqinline{Monoid.type} rather than \coqinline{Type} as its argument.

\begin{coqcode}
Module Semiring.

Record mixin_of (A : Monoid.type) := Mixin {
  one : A;
  mul : A -> A -> A;
  addC : commutative (@add A);          (* `add' is commutative.              *)
  mulA : associative mul;               (* `mul' is associative.              *)
  mul1x : left_id one mul;              (* `one' is the left and right        *)
  mulx1 : right_id one mul;             (*   identity element w.r.t. `mul'.   *)
  mulDl : left_distributive mul add;    (* `mul' is left and right            *)
  mulDr : right_distributive mul add;   (*   distributive over `add'.         *)
  mul0x : left_zero zero mul;           (* `zero' is the left and right       *)
  mulx0 : right_zero zero mul;          (*   absorbing element w.r.t. `mul'.  *)
}.
\end{coqcode}
The \coqinline{Semiring} class packs the \coqinline{Semiring} mixin together with the \coqinline{Monoid} class to assemble the mixin records of monoids and semirings.
We may also assemble all the required mixins as record fields directly rather than nesting class records, yielding what is called the \emph{flat} variant of packed classes~\cite[Sect.~4]{hierarchy-builder}.
Since the semiring mixin requires \coqinline{Monoid.type} as its type argument, we have to bundle the monoid class with the carrier to provide that \coqinline{Monoid.type} instance, as follows.
\begin{coqcode}
Record class_of (A : Type) :=
  Class { base : Monoid.class_of A; mixin : mixin_of |*(Monoid.Pack A base)*| }.

Structure type := Pack { sort : Type; class : class_of sort }.
\end{coqcode}
The inheritance from monoids to semirings can then be expressed as a canonical way to construct a monoid from a semiring as below.
\begin{coqcode}
Local Definition monoidType (cT : type) : Monoid.type :=
  Monoid.Pack (sort cT) (base _ (class cT)).

End Semiring.
\end{coqcode}

Following the above method, we declare \coqinline{Semiring.sort} as an implicit coercion, and then lift \coqinline{mul}, \coqinline{one}, and the semiring axioms from projections for the mixin to definitions for \coqinline{Semiring.type}.
\begin{coqcode}[numbers=left, alsoletter={.}, morekeywords={[7]Semiring.monoidType}, keywordstyle={[7]\ttfamily\color{commentcolor}}]
Coercion Semiring.sort : Semiring.type >-> Sortclass.
Definition one {A : Semiring.type} : A :=
  Semiring.one _ (Semiring.mixin _ (Semiring.class A)).
Definition mul {A : Semiring.type} : A -> A -> A :=
  Semiring.mul _ (Semiring.mixin _ (Semiring.class A)).
Lemma addC {A : Semiring.type} : commutative (@add (Semiring.monoidType A)). [*\label{line:addC}*]
...
\end{coqcode}
In the statement of the \coqinline{addC} axiom (line \ref{line:addC} just above), we need to explicitly write \coqinline{Semiring.monoidType A} to get the canonical \coqinline{Monoid.type} instance for \hbox{\coqinline{A : Semiring.type}}.
We omit this subtyping function \coqinline{Semiring.monoidType} by declaring it as an implicit coercion.
In general, for a structure \coqinline{S} inheriting from other structures, we define implicit coercions from \coqinline{S} to all its superclasses.
\begin{coqcode}
Coercion Semiring.monoidType : Semiring.type >-> Monoid.type.
\end{coqcode}

The \coqinline{Group} structure is monoids extended with an additive inverse.
Following the above method, it can be defined as follows.
\begin{coqcode}
Module Group.

Record mixin_of (A : Monoid.type) := Mixin {
  opp : A -> A;
  addNx : left_inverse zero opp add;    (* `opp x' is the left and right      *)
  addxN : right_inverse zero opp add;   (*   additive inverse of `x'.         *)
}.

Record class_of (A : Type) :=
  Class { base : Monoid.class_of A; mixin : mixin_of (Monoid.Pack A base) }.

Structure type := Pack { sort : Type; class : class_of sort }.

Local Definition monoidType (cT : type) : Monoid.type :=
  Monoid.Pack (sort cT) (base _ (class cT)).

End Group.

Coercion Group.sort : Group.type >-> Sortclass.
Coercion Group.monoidType : Group.type >-> Monoid.type.
Definition opp {A : Group.type} : A -> A :=
  Group.opp _ (Group.mixin _ (Group.class A)).
...
\end{coqcode}

The \coqinline{Ring} structure can be seen both as groups extended by the semiring axioms and as semirings extended by the group axioms.
Here, we define it in the first way, but one may also define it in the second way.
Since rings have no other axioms than the group and semiring axioms, no additional \coqinline{mixin_of} record is needed.\footnote{One may also define a new structure that inherits from multiple existing classes and has an extra mixin, e.g., by defining commutative rings instead of rings in this example, and left algebras as \coqinline{lalgType} in \MC.}
\begin{coqcode}
Module Ring.

Record class_of (A : Type) := Class {
  base : Group.class_of A;
  mixin : Semiring.mixin_of (Monoid.Pack A (Group.base A base)) }.

Structure type := Pack { sort : Type; class : class_of sort }.
\end{coqcode}
The ring structure inherits from monoids, groups, and semirings.
Here, we define implicit coercions from the ring structure to those superclasses.
\begin{coqcode}
Local Definition monoidType (cT : type) : Monoid.type :=
  Monoid.Pack (sort cT) (Group.base _ (base _ (class cT))).
Local Definition groupType (cT : type) : Group.type :=
  Group.Pack (sort cT) (base _ (class cT)).
Local Definition semiringType (cT : type) : Semiring.type :=
  Semiring.Pack (sort cT) (Semiring.Class _ (Group.base _ (base _ (class cT)))
                                          (mixin _ (class cT))).

End Ring.

Coercion Ring.sort : Ring.type >-> Sortclass.
Coercion Ring.monoidType : Ring.type >-> Monoid.type.
Coercion Ring.semiringType : Ring.type >-> Semiring.type.
Coercion Ring.groupType : Ring.type >-> Group.type.
\end{coqcode}


\pagebreak

\section{Coherence of implicit coercions}
\label{sec:coherence}

This section describes the implicit coercion mechanism of \Coq{} and the \emph{coherence property}~\cite{Barthe:1995} of inheritance graphs that ensures modularity of reasoning with packed classes, and presents the coherence checking mechanism we implemented in \Coq.
In Appendix \ref{sec:invariant-concrete-modularity}, we study another invariant of hierarchies concerning concrete instances and implicit coercions.
More details on implicit coercions can be found in the \Coq{} reference manual~\cite{coqrefman-coercions}, and its typing algorithm is described in~\cite{Saibi:1997}.
First, we define classes and implicit coercions.

\begin{definition}[{Classes~\cite[Sect.~3.1]{Saibi:1997}\cite{coqrefman-coercions}}]
 \label{def:coercion-classes}
 A class with $n$ parameters is a defined name $C$ with a type $\forall (x_1 : T_1) \dots (x_n : T_n), \mathit{sort}$ where $\mathit{sort}$ is \coqinline{SProp}, \coqinline{Prop}, \coqinline{Set}, or \coqinline{Type}.
 Thus, a class with parameters is considered a single class and not a family of classes.
 An object of class $C$ is any term of type $C \, t_1 \dots t_n$.
\end{definition}

\begin{definition}[Implicit coercions]
 \label{def:implicit-coercions}
 A name $f$ can be declared as an implicit coercion from a source class $C$ to a target class $D$ with $k$ parameters if the type of $f$ has the form $\forall x_1 \dots x_k \, (y : C \, t_1 \dots t_n), D \, u_1 \dots u_m$. We then write $f : C \rightarrowtail D$.\footnote{In fact, the target classes can also be functions (\coqinline{Funclass}) and sorts (\coqinline{Sortclass}); that is to say, a function returning functions, types, or propositions can be declared as an implicit coercion. In this paper, we omit these cases to simplify the presentation, but our discussion can be generalized to these cases.}
\end{definition}

An implicit coercion $f : C \rightarrowtail D$ can be seen as a subtyping $C \leq D$ and applied to fill type mismatches to a term of class $C$ placed in a context that expects to have a term of class $D$.
Implicit coercions form an inheritance graph with classes as nodes and coercions as edges, whose path $[f_1; \dots; f_n]$ where $f_i : C_i \rightarrowtail C_{i + 1}$ can also be seen as a subtyping $C_1 \leq C_{n + 1}$; thus, we write $[f_1; \dots; f_n] : C_1 \rightarrowtail C_{n + 1}$ to indicate $[f_1; \dots; f_n]$ is an inheritance path from $C_1$ to $C_{n + 1}$.
The \Coq{} system pre-computes those inheritance paths for any pair of source and target classes, and updates to keep them closed under transitivity when a new implicit coercion is declared~\cite[Sect.~3.3]{Saibi:1997}\cite[Sect.~8.2.5 ``Inheritance Graph'']{coqrefman-coercions}.
The coherence of inheritance graphs is defined as follows.

\begin{definition}[Definitional equality~{\cite{coqrefman-conversion}\cite[Sect.~3.1]{Coquand:1988}}]
Two terms $t_1$ and $t_2$ are said to be definitionally equal, or convertible, if they are equivalent under $\beta\delta\iota\zeta$-reduction and $\eta$-expansion.
 This equality is denoted by the infix symbol $\conv$.
\end{definition}

\begin{definition}[{Coherence~\cite[Sect.~3.2]{Barthe:1995}\cite[Sect.~7]{Saibi:1997}}]
 \label{def:coherence}
 An inheritance graph is coherent if and only if the following two conditions hold.
 \begin{enumerate}
  \item For any circular inheritance path $p : C \rightarrowtail C$, $p \, x \conv x$, where $x$ is a fresh variable of class $C$.
  \item For any two inheritance paths $p, q : C \rightarrowtail D$, $p \, x \conv q \, x$, where $x$ is a fresh variable of class $C$.
 \end{enumerate}
\end{definition}

\pagebreak Before our work, if multiple inheritance paths existed between the same source and target class, only the oldest one was kept as a valid one in the inheritance graph in \Coq{}, and all the others were reported as \emph{ambiguous paths} and ignored.
We improved this mechanism to report only paths that break the coherence conditions and also to minimize the number of reported ambiguous paths~\cite{Coq:PR9743,Coq:PR11258}. 
The second condition ensures the modularity of reasoning with packed classes. For example, proving $\forall (R : \text{\coqinline{Ring.type}}) \, (x, y : R), (- x) \times y = - (x \times y)$ requires using both
\coqinline[alsoletter={.}, morekeywords={[7]Semiring.monoidType,Ring.semiringType}, keywordstyle={[7]\ttfamily\color{commentcolor}}]{Semiring.monoidType (Ring.semiringType R)} and
\coqinline[alsoletter={.}, morekeywords={[7]Group.monoidType,Ring.groupType}, keywordstyle={[7]\ttfamily\color{commentcolor}}]{Group.monoidType (Ring.groupType R)}
implicitly. If those \coqinline{Monoid} instances are not definitionally equal, it will prevent us from proving the lemma by reporting type mismatch between \coqinline{R} and \coqinline{R}.

Convertibility checking for inheritance paths consisting of implicit coercions as in Definition \ref{def:implicit-coercions} requires constructing a composition of functions for a given inheritance path.
One can reduce any unification problem to this well-typed term construction problem, that in the higher-order case is undecidable~\cite{Goldfarb:1981}.
However, the inheritance paths that make the convertibility checking undecidable can never be applied as implicit coercions in type inference, because they do not respect the uniform inheritance condition.

\begin{definition}[{Uniform inheritance condition~\cite{coqrefman-coercions}\cite[Sect.~3.2]{Saibi:1997}}]
 \label{def:uniform-inheritance-condition}
 An implicit coercion $f$ between classes $C \rightarrowtail D$ with $n$ and $m$ parameters, respectively, is uniform if and only if the type of $f$ has the form
 \[
  \forall (x_1 : A_1) \dots (x_n : A_n) \, (x_{n + 1} : C \, x_1 \dots x_n), D \, u_1 \dots u_m.
 \]
\end{definition}

\begin{remark}
 Names that can be declared as implicit coercions are defined as constants that respect the uniform inheritance condition in~\cite[Sect.~3.2]{Saibi:1997}. However, the actual implementation in the modern \Coq{} system accepts almost any function as in Definition \ref{def:implicit-coercions} as a coercion.
\end{remark}

Sa\"{\i}bi claimed that the uniform inheritance condition ``ensures that any coercion can be applied to any object of its source class''~\cite[Sect.~3.2]{Saibi:1997}, but the actual condition ensures additional properties.
The number and ordering of parameters of a uniform implicit coercion are the same as those of its source class; thus, convertibility checking of uniform implicit coercions $f, g : C \rightarrowtail D$ does not require any special treatment such as permuting parameters of $f$ and $g$.
Moreover, function composition preserves this uniformity, that is, the following lemma holds.

\begin{lemma}
 \label{lem:uniform-inheritance-composition}
 For any uniform implicit coercions $f : C \rightarrowtail D$ and $g : D \rightarrowtail E$, the function composition of the inheritance path $[f; g] : C \rightarrowtail E$ is uniform.
\end{lemma}

\begin{proof}
 Let us assume that $C$, $D$, and $E$ are classes with $n$, $m$, and $k$ parameters respectively, and $f$ and $g$ have the following types:
 \begin{align*}
  f &: \forall (x_1 : T_1) \dots (x_n : T_n) \, (x_{n + 1} : C \, x_1 \dots x_n), D \, u_1 \dots u_m, \\
  g &: \forall (y_1 : U_1) \dots (y_m : U_m) \, (y_{m + 1} : D \, y_1 \dots y_m), E \, v_1 \dots v_k.
 \end{align*}
 \pagebreak

 \noindent Then, the function composition of $f$ and $g$ can be defined and typed as follows:
 \begin{align*}
  g \circ f
  &:= \lambda (x_1 : T_1) \dots (x_n : T_n) \, (x_{n + 1} : C \, x_1 \dots x_n), g \, u_1 \dots u_m \, (f \, x_1 \dots x_n \, x_{n + 1}) \\
  &: \forall (x_1 : T_1) \dots (x_n : T_n) \, (x_{n + 1} : C \, x_1 \dots x_n), \\
  &\phantom{{}:{}} E \, (v_1 \{y_1 / u_1\} \dots \{y_m / u_m\} \{y_{m + 1} / f \, x_1 \dots x_n \, x_{n + 1}\}) \\
  &\phantom{{}: E \,\,\,\,} \smash{\vdots} \\
  &\phantom{{}: E \,} (v_k \{y_1 / u_1\} \dots \{y_m / u_m\} \{y_{m + 1} / f \, x_1 \dots x_n \, x_{n + 1}\}).
 \end{align*}
 The terms $u_1, \dots, u_m$ contain the free variables $x_1, \dots, x_{n + 1}$ and we omitted substitutions for them by using the same names for the binders in the above definition.
 Nevertheless, $(g \circ f) : C \rightarrowtail E$ respects the uniform inheritance condition. \qedhere
\end{proof}

In the above definition of the function composition $g \circ f$ of implicit coercions, the types of $x_1, \dots, x_n, x_{n + 1}$ and the parameters of $g$ can be automatically inferred in \Coq; thus, it can be abbreviated as follows:
\[
 g \circ f := \lambda (x_1 : \_) \dots (x_n : \_) \, (x_{n + 1} : \_), g \, \underbrace{\_ \, \dots \, \_}_{\mathclap{\text{$m$ parameters}}} \, (f \, x_1 \dots x_n \, x_{n + 1}).
\]

For implicit coercions $f_1 : C_1 \rightarrowtail C_2, f_2 : C_2 \rightarrowtail C_3, \dots, f_n : C_n \rightarrowtail C_{n + 1}$ that have $m_1, m_2, \dots, m_n$ parameters respectively, the function composition of the inheritance path $[f_1; f_2; \dots; f_n]$ can be written as follows by repeatedly applying Lemma \ref{lem:uniform-inheritance-composition} and the above abbreviation.
\begin{align*}
 f_n \circ \dots \circ f_2 \circ f_1 :={}
 &\lambda (x_1 : \_) \dots (x_{m_1} : \_) \, (x_{m_1 + 1} : \_), \\
 &\qquad f_n \, \underbrace{\_ \, \dots \, \_}_{\mathclap{\text{$m_n$ parameters}}} \,
         (\dots(f_2 \, \underbrace{\_ \, \dots \, \_}_{\mathclap{\text{$m_2$ parameters}}} \, (f_1 \, x_1 \dots x_{m_1} \, x_{m_1 + 1}))\dots).
\end{align*}
If $f_1, \dots, f_n$ are all uniform, the numbers of their parameters $m_1, \dots, m_n$ are equal to the numbers of parameters of $C_1, \dots, C_n$.
Consequently, the type inference algorithm always produces the typed closed term of most general function composition of $f_1, \dots, f_n$ from the above term.
If not all of $f_1, \dots, f_n$ are uniform, type inference may fail or produce an open term, but if this produces a typed closed term, it is the most general function composition of $f_1, \dots, f_n$.

Our coherence checking mechanism constructs the function composition of $p : C \rightarrowtail C$ and compares it with the identity function of class $C$ to check the first condition, and also constructs the function compositions of $p, q : C \rightarrowtail D$ and performs the conversion test for them to check the second condition.

\section{Automated structure inference}
\label{sec:automated-structure-inference}

This section reviews how the automated structure inference mechanism~\cite{Mahboubi:2013} works on our example and in general.
The first example is $0 + 1$, whose desugared form is \coqinline{@add _ (@zero _) (@one _)}, where holes \coqinline{_} stand for implicit pieces of information to be inferred.
The left- and right-hand sides of the top application can be type-checked without any use of canonical structures, as follows:
\begin{coqcode}
$\evar{M}$ : Monoid.type |- @add $\evar{M}$ (@zero $\evar{M}$) : Monoid.sort $\evar{M}$ -> Monoid.sort $\evar{M}$,
$\evar{SR}$ : Semiring.type |- @one $\evar{SR}$ : Semiring.sort $\evar{SR}$,
\end{coqcode}
\pagebreak where $\evar{M}$ and $\evar{SR}$ represent unification variables.
Type-checking the application requires solving a unification problem $\text{\coqinline{Monoid.sort}} ~ \evar{M} \unify \text{\coqinline{Semiring.sort}} ~ \evar{SR}$, which is not trivial and which \Coq{} does not know how to solve without hints.
By declaring \coqinline{Semiring.monoidType : Semiring.type -> Monoid.type} as a canonical instance, \Coq{} can become aware of that instantiating $\evar{M}$ with \linebreak \coqinline{Semiring.monoidType $\evar{SR}$} is the canonical solution to this unification problem.
\begin{coqcode}
Canonical Semiring.monoidType.
\end{coqcode}
The \coqinline{Canonical} command takes a definition with a body of the form $\lambda x_1 \, \dots \, x_n,$ $\{|p_1 := (f_1 \dots); \dots; p_m := (f_m \dots)|\}$, and then synthesizes unification hints between the projections $p_1, \dots, p_m$ and the head symbols $f_1, \dots, f_m$, respectively, except for unnamed projections.
Since \coqinline{Semiring.monoidType} has the following body, the above \coqinline{Canonical} declaration synthesizes the unification hint between \coqinline{Monoid.sort} and \coqinline{Semiring.sort} that we need:
\begin{coqcode}
fun cT : Semiring.type =>
  {| |*Monoid.sort*| := |*Semiring.sort*| cT;
     Monoid.class := Semiring.base cT (Semiring.class cT) |}.
\end{coqcode}

In general, for any structures \coqinline{A} and \coqinline{B} such that \coqinline{B} inherits from \coqinline{A} with an implicit coercion \coqinline{B.aType : B.type >-> A.type}, \coqinline{B.aType} should be declared as a canonical instance to allow \Coq{} to solve unification problems of the form $\text{\coqinline{A.sort}} ~ \evar{A} \unify \text{\coqinline{B.sort}} ~ \evar{B}$ by instantiating $\evar{A}$ with \coqinline{B.aType $\evar{B}$}.

The second example is $- 1$, whose desugared form is \coqinline{@opp _ (@one _)}.
The left- and right-hand sides of the top application can be type-checked as follows:
\begin{coqcode}
$\evar{G}$ : Group.type |- @opp $\evar{G}$ : Group.sort $\evar{G}$ -> Group.sort $\evar{G}$,
$\evar{SR}$ : Semiring.type |- @one $\evar{SR}$ : Semiring.sort $\evar{SR}$.
\end{coqcode}
In order to type check the application, \Coq{} has to unify $\text{\coqinline{Group.sort}} ~ \evar{G}$ with $\text{\coqinline{Semiring.sort}} ~ \evar{SR}$, which, again, is not trivial.
Moreover, groups and semirings do not inherit from each other; therefore, this case is not an instance of the above criteria to define canonical instances.
Nevertheless, this unification problem means that \coqinline{$\evar{G}$ : Group.type} and \coqinline{$\evar{SR}$ : Semiring.type} are the same, and they are equipped with both group and semiring axioms, that is, the ring structure. Thus, its canonical solution should be introducing a fresh unification variable \coqinline{$\evar{R}$ : Ring.type} and instantiating $\evar{G}$ and $\evar{SR}$ with \coqinline{Ring.groupType $\evar{R}$} and \coqinline{Ring.semiringType $\evar{R}$}, respectively.
Right after defining \coqinline{Ring.semiringType}, this unification hint can be defined as follows.

\begin{coqcode}
Local Definition semiring_groupType (cT : type) : Group.type :=
  Group.Pack (|*Semiring.sort*| (semiringType cT)) (base _ (class cT)).
\end{coqcode}
This definition is definitionally equal to \coqinline{Ring.groupType}, but has a different head symbol, \coqinline{Semiring.sort} instead of \coqinline{Ring.sort}, in its first field \coqinline{Group.sort}.
Thus, the unification hint we need between \coqinline{Group.sort} and \coqinline{Semiring.sort} can be synthesized by the following declarations.
\begin{coqcode}
Canonical Ring.semiring_groupType.
\end{coqcode}
This unification hint can also be defined conversely as follows.
Whichever of those is acceptable, but at least one of them should be declared.
\begin{coqcode}
Local Definition group_semiringType (cT : type) : Semiring.type :=
  Semiring.Pack (Group.sort (groupType cT))
    (Semiring.Class _ (Group.base _ (base _ (class cT))) (mixin _ (class cT))).
\end{coqcode}

\pagebreak For any structures $A$ and $B$ that have common (non-strict) subclasses $\mathcal{C}$, we say that $C \in \mathcal{C}$ is a \emph{join} of $A$ and $B$ if $C$ does not inherit from any other structures in $\mathcal{C}$.
For example, if we add the structure of commutative rings to the hierarchy of Sect.~\ref{sec:packed-classes}, the commutative ring structure is a common subclass of the group and semiring structures, but is not a join of them because the commutative ring structure inherits from the ring structure which is another common subclass of them.
In general, the join of any two structures must be unique, and we should declare a canonical instance to infer the join $C$ as the canonical solution of unification problems of the form $\text{\coqinline{A.sort}} ~ \evar{A} \unify \text{\coqinline{B.sort}} ~ \evar{B}$.
For any structures $A$ and $B$ such that $B$ inherits from $A$, $B$ is the join of $A$ and $B$; thus, the first criteria to define canonical instances is just an instance of the second criteria.

\begin{figure}[t]
 \newlength{\figwidth}
 \setlength{\figwidth}{\linewidth}
 \addtolength{\figwidth}{-\columnsep}
 \setlength{\figwidth}{.5\figwidth}
 \begin{minipage}[t]{\figwidth}
  \centering
  \begin{tikzpicture}[structure/.style={draw, rounded corners=1mm}, x=15mm, y=7mm]
   \node[structure] (Type) at (0,  3) {\coqinline{Type}};
   \node[structure] (A)    at (-1, 2) {\coqinline{A.type}};
   \node[structure] (B)    at (1,  2) {\coqinline{B.type}};
   \node[structure] (C)    at (-1, 0) {\coqinline{C.type}};
   \node[structure] (D)    at (1,  0) {\coqinline{D.type}};
   \draw[->, densely dotted] (Type) -- (A);
   \draw[->, densely dotted] (Type) -- (B);
   \draw[->] (A) -- (C);
   \draw[->] (A) -- (D);
   \draw[->] (B) -- (C);
   \draw[->] (B) -- (D);
  \end{tikzpicture}
  \caption{A minimal hierarchy that has ambiguous joins. Both structure \coqinline{C} and \coqinline{D} directly inherit from the structures \coqinline{A} and \coqinline{B}; thus, \coqinline{A} and \coqinline{B} have two joins.}
  \label{fig:ambiguous-hierarchy}
 \end{minipage}\hspace{\columnsep}%
 \begin{minipage}[t]{\figwidth}
  \centering
  \begin{tikzpicture}[structure/.style={draw, rounded corners=1mm}, x=15mm, y=7mm]
   \node[structure] (Type) at (0,  3) {\coqinline{Type}};
   \node[structure] (A)    at (-1, 2) {\coqinline{A.type}};
   \node[structure] (B)    at (1,  2) {\coqinline{B.type}};
   \node[structure] (Join) at (0,  1) {\coqinline{join(A, B).type}};
   \node[structure] (C)    at (-1, 0) {\coqinline{C.type}};
   \node[structure] (D)    at (1,  0) {\coqinline{D.type}};
   \draw[->, densely dotted] (Type) -- (A);
   \draw[->, densely dotted] (Type) -- (B);
   \draw[->] (A) -- (Join);
   \draw[->] (B) -- (Join);
   \draw[->] (Join) -- (C);
   \draw[->] (Join) -- (D);
  \end{tikzpicture}
  \caption{A hierachy that disambiguates the join of \coqinline{A} and \coqinline{B} in Fig.~\ref{fig:ambiguous-hierarchy} by redefining \coqinline{C} and \coqinline{D} to inherit from a new structure \coqinline{join(A, B)} that inherits from \coqinline{A} and \coqinline{B}.}
  \label{fig:disambiguated-hierarchy}
 \end{minipage}
\end{figure}

Fig.~\ref{fig:ambiguous-hierarchy} shows a minimal hierarchy that has ambiguous joins.
If we declare that \coqinline{C} (resp.~\coqinline{D}) is the canonical join of \coqinline{A} and \coqinline{B} in this hierarchy, it will also be accidentally inferred for a user who wants to reason about \coqinline{D} (resp.~\coqinline{C}).
Since \coqinline{C} and \coqinline{D} do not inherit from each other, inferred \coqinline{C} (resp.~\coqinline{D}) can never be instantiated with \coqinline{D} (resp.~\coqinline{C}); therefore, we have to disambiguate it as in Fig.~\ref{fig:disambiguated-hierarchy}, so that the join of \coqinline{A} and \coqinline{B} can be specialized to both \coqinline{C} and \coqinline{D} afterwards.

\section{A simplified formal model of hierarchies}
\label{sec:formal-hierarchy}

In this section, we define a simplified formal model of hierarchies and show a metatheorem that ensures the predictability of structure inference. First, we define the model of hierarchies and inheritance relations.

\begin{definition}[Hierarchy and inheritance relations]
 \label{def:hierarchy}
 A hierarchy $\mathcal{H}$ is a finite set of structures partially ordered by a non-strict inheritance relation $\leadsto^*$; that is, $\leadsto^*$ is reflexive, antisymmetric, and transitive.
 We denote the corresponding strict (irreflexive) inheritance relation by $\leadsto^+$. $A \leadsto^* B$ and $A \leadsto^+ B$ respectively mean that $B$ non-strictly and strictly inherits from $A$.
\end{definition}


\begin{definition}[Common subclasses]
 \label{def:common-subclasses}
 The (non-strict) common subclasses of $A, B \in \mathcal{H}$ are $\mathcal{C} := \{C \in \mathcal{H} \mid A \leadsto^* C \land B \leadsto^* C\}$.
 The minimal common subclasses of $A$ and $B$ is $\mathrm{mcs}(A, B) := \mathcal{C} \setminus \{C \in \mathcal{H} \mid \exists C' \in \mathcal{C}, C' \leadsto^+ C\}$.
\end{definition}

\pagebreak

\begin{definition}[Well-formed hierarchy]
 \label{def:well-formed-hierarchy}
 A hierarchy $\mathcal{H}$ is said to be well-formed if the minimal common subclasses of any two structures are unique; that is, $|\mathrm{mcs}(A, B)| \leq 1$ for any $A, B \in \mathcal{H}$.
\end{definition}

\begin{definition}[Extended hierarchy]
 \label{def:extended-hierarchy}
 An extended hierarchy $\bar{\mathcal{H}} := \mathcal{H} \mathbin{\dot{\cup}} \{\top\}$ is a hierarchy $\mathcal{H}$ extended with $\top$ which means a structure that strictly inherits from all the structures in $\mathcal{H}$; thus, the inheritance relation is extended as follows:
 \begin{align*}
  A \mathrel{\bar{\leadsto}}^* \top &\iff \mathrm{true}, \\
  \top \mathrel{\bar{\leadsto}}^* B &\iff \mathrm{false} & \text{(if $B \neq \top$)}, \\
  A \mathrel{\bar{\leadsto}}^* B    &\iff A \leadsto^* B & \text{(if $A \neq \top$ and $B \neq \top$)}.
 \end{align*}
\end{definition}

\begin{definition}[Join]
 \label{def:unique-join}
 The join is a binary operator on an extended well-formed hierarchy $\bar{\mathcal{H}}$, defined as follows:
 \[
  \mathrm{join}(A, B) =
  \begin{cases}
   C    & \text{(if $A, B \in \mathcal{H}$ and $\mathrm{mcs}(A, B) = \{C\}$)}, \\
   \top & \text{(otherwise)}.
  \end{cases}
 \]
\end{definition}

We encoded the above definitions on hierarchies in \Coq{} by using the structure of partially ordered finite types \coqinline{finPOrderType} of the \texttt{mathcomp-finmap} library~\cite{mathcomp-finmap-github} and proved the following theorem.

\begin{theorem}
 \label{thm:join-AC}
 The join operator on an extended well-formed hierarchy is associative, commutative, and idempotent; that is, an extended well-formed hierarchy is a join-semilattice.
\end{theorem}

\begin{proof}
 See Appendix \ref{sec:formal-hierarchy-proofs}. \qedhere
\end{proof}

If the unification algorithm of \Coq{} respects our model of hierarchies and joins during structure inference, Theorem \ref{thm:join-AC} implies that permuting, duplicating, and contracting unification problems do not change the result of inference; thus, it states the predictability of structure inference at a very abstract level.

\section{Validating well-formedness of hierarchies}
\label{sec:validating-canonical-projections}

This section presents a well-formedness checking algorithm that can also generate the exhaustive set of assertions for joins.
We implemented this checking mechanism as a tool \texttt{hierarchy.ml} written in \OCaml, which is available as a \MC{} developer utility~\cite[\texttt{/etc/utils/hierarchy.ml}]{mathcomp-github}.
Our checking tool outputs the assertions as a \Coq{} proof script which can detect missing and misimplemented unification hints for joins.

\begin{algorithm}[t]
 \DontPrintSemicolon
 \SetKwInOut{Params}{Parameters}
 \SetKwFunction{Join}{join}
 \SetKwFunction{Sub}{subof}
 \SetKwProg{Fn}{Function}{:}{}
 \SetKw{Fail}{fail}
 \SetKw{Print}{print}

 \Params{$\mathcal{H}$ is the set of all the structures. \Sub is the map from structures to their strict subclasses, which is required to be transitive: $\forall A \, B \in \mathcal{H}, A \in \Sub(B) \Rightarrow \Sub(A) \subset \Sub(B)$.}

 \BlankLine
 \Fn{\Join{$A$, $B$}}{
  $\mathcal{C} := (\Sub{$A$} \cup \{A\}) \cap (\Sub{$B$} \cup \{B\})$;

  \tcc*[f]{$\mathcal{C}$ is the set of all the common subclasses of $A$ and $B$.}

  \lForEach{$C \in \mathcal{C}$}{$\mathcal{C} \leftarrow \mathcal{C} \setminus \Sub{$C$}$;}
  \tcc*[f]{Since \Sub is transitive, removed elements of $\mathcal{C}$ can be skipped in this loop, and the ordering of picking elements from $\mathcal{C}$ does not matter.}

  \lIf{$\mathcal{C} = \emptyset$}{\Return{$\top$}; \tcc*[f]{There is no join of $A$ and $B$.}}
  \lElseIf{$\mathcal{C}$ is singleton $\{C\}$}{\Return{$C$}; \tcc*[f]{$C$ is the join of $A$ and $B$.}}
  \lElse{\Fail; \tcc*[f]{The join of $A$ and $B$ is ambiguous.}}
 }

 \BlankLine
 \ForEach{$A \in \mathcal{H}, B \in \mathcal{H}$}{
  $C := \Join{$A$, $B$}$;
  \lIf{$A \neq B \land C \neq \top$}{\Print ``\texttt{check\_join} $A$ $B$ $C$.'';}
 }
 \caption{Well-formedness checking and assertions generation}
 \label{alg:hierarchy-checking}
\end{algorithm}

Since a hierarchy must be a finite set of structures (Definition \ref{def:hierarchy}), Definitions \ref{def:common-subclasses}, \ref{def:unique-join}, and \ref{def:well-formed-hierarchy} give us computable (yet inefficient) descriptions of joins and the well-formedness; in other words, for a given hierarchy $\mathcal{H}$ and any two structures $\text{\coqinline{A}}, \text{\coqinline{B}} \in \mathcal{H}$, one may enumerate their minimal common subclasses.
Algorithm \ref{alg:hierarchy-checking} is the checking algorithm we present, that takes in input an inheritance relation in the form of an indexed family of strict subclasses $\texttt{subof}(A) := \{B \in \mathcal{H} \mid A \leadsto^+ B\}$.
The \texttt{join} function in this algorithm takes two structures as arguments, checks the uniqueness of their join, and then returns the join if it uniquely exists. \pagebreak
In this function, the enumeration of minimal common subclasses is done by constructing the set of common subclasses $\mathcal{C}$, and filtering out $\texttt{subof}(C)$ from $\mathcal{C}$ for every $C \in \mathcal{C}$.
In this filtering process, which is written as a \textbf{foreach} statement, we can skip elements already filtered out and do not need to care about ordering of picking up elements, thanks to transitivity.

The \texttt{hierarchy.ml} utility extracts the inheritance relation from a \Coq{} library by interacting with \texttt{coqtop}, and then executes Algorithm \ref{alg:hierarchy-checking} to check the well-formedness and to generate assertions.
The assertions generated from our running example are shown below.
\begin{coqcode}[numbers=left]
check_join Group.type Monoid.type Group.type. [*\label{line:check_join-group-monoid}*]
check_join Group.type Ring.type Ring.type. [*\label{line:check_join-group-ring}*]
check_join Group.type Semiring.type Ring.type. [*\label{line:check_join-group-semiring}*]
check_join Monoid.type Group.type Group.type. [*\label{line:check_join-monoid-group}*]
check_join Monoid.type Ring.type Ring.type. [*\label{line:check_join-monoid-ring}*]
check_join Monoid.type Semiring.type Semiring.type. [*\label{line:check_join-monoid-semiring}*]
check_join Ring.type Group.type Ring.type. [*\label{line:check_join-ring-group}*]
check_join Ring.type Monoid.type Ring.type. [*\label{line:check_join-ring-monoid}*]
check_join Ring.type Semiring.type Ring.type. [*\label{line:check_join-ring-semiring}*]
check_join Semiring.type Group.type Ring.type. [*\label{line:check_join-semiring-group}*]
check_join Semiring.type Monoid.type Semiring.type. [*\label{line:check_join-semiring-monoid}*]
check_join Semiring.type Ring.type Ring.type. [*\label{line:check_join-semiring-ring}*]
\end{coqcode}
An assertion \coqinline{check_join t1 t2 t3} asserts that the join of \coqinline{t1} and \coqinline{t2} is \coqinline{t3}, and \coqinline{check_join} is implemented as a tactic that fails if the assertion is false.
For instance, if we do not declare \coqinline{Ring.semiringType} as a canonical instance, the assertion of line \ref{line:check_join-ring-semiring} fails and reports the following error.
\begin{code}
There is no join of Ring.type and Semiring.type but it is expected to be
Ring.type.
\end{code}

One may declare incorrect canonical instances that overwrite an existing join.
For example, the join of groups and monoids must be groups; however, defining the following canonical instance in the \coqinline{Ring} section overwrites this join.
\begin{coqcode}
Local Definition bad_monoid_groupType : Group.type :=
  Group.Pack (Monoid.sort monoidType) (base _ class).
\end{coqcode}
\pagebreak By declaring \coqinline{Ring.bad_monoid_groupType} as a canonical instance, the join of \coqinline{Monoid.type} and \coqinline{Group.type} is still \coqinline{Group.type}, but the join of \coqinline{Group.type} and \coqinline{Monoid.type} becomes \coqinline{Ring.type}, because of asymmetry of the unification mechanism.
The assertion of line \ref{line:check_join-group-monoid} fails and reports the following error.
\begin{code}
The join of Group.type and Monoid.type is Ring.type but it is expected to be
Group.type.
\end{code}

\section{Evaluation}
\label{sec:evaluation}

This section reports the results of applying our tools to the \MC{} library 1.7.0 and on recent development efforts to extend the hierarchy in \MC{} using our tools.
For further details, see Appendix \ref{sec:evolution-mathcomp}.
\MC{} 1.7.0 provides the structures depicted in Fig.~\ref{fig:hierarchy}, except \coqinline{comAlgType} and \coqinline{comUnitAlgType}, and lacked a few of the edges; thus, its hierarchy is quite large.
Our coherence checking mechanism found 11 inconvertible multiple inheritance paths in the \textsf{ssralg} library.
Fortunately, those paths concern proof terms and are intended to be irrelevant; hence, we can ensure no implicit coercion breaks the modularity of reasoning.
Our well-formedness checking tool discovered 7 ambiguous joins, 8 missing unification hints, and one overwritten join due to an incorrect declaration of a canonical instance.
These inheritance bugs were found and fixed with the help of our tools; thus, similar issues cannot be found in the later versions of \MC{}.

The first issue was that inheritance from the \coqinline{CountRing} structures \linebreak (\coqinline{countZmodType} and its subclasses with the prefix \coqinline{count}) to the \coqinline{FinRing} structures (\coqinline{finZmodType} and its subclasses) was not implemented, and consequently it introduced 7 ambiguous joins.
For instance, \coqinline{finZmodType} did not inherit from \coqinline{countZmodType} as it should; consequently, they became ambiguous joins of \coqinline{countType} and \coqinline{zmodType}.
6 out of 8 missing unification hints should infer \coqinline{CountRing} or \coqinline{FinRing} structures.
The other 2 unification hints are in numeric field (\coqinline{numFieldType} and \coqinline{realFieldType}) structures.
Fixing the issue of missing inheritance from \coqinline{CountRing} to \coqinline{FinRing} was a difficult task without tool support.
The missing inheritance itself was a known issue from before our tooling work, but the sub-hierarchy consisting of the \coqinline{GRing}, \coqinline{CountRing}, and \coqinline{FinRing} structures in Fig.~\ref{fig:hierarchy} is quite dense; as a result, it prevents the library developers from enumerating joins correctly without automation~\cite{MathComp:PR291}.

The second issue was that the following canonical \coqinline{finType} instance for \coqinline{extremal_group} overwrote the join of \coqinline{finType} and \coqinline{countType}, which should be \coqinline{finType}.
\begin{coqcode}
Canonical extremal_group_finType := FinType _ extremal_group_finMixin.
\end{coqcode}
In this declaration, \coqinline{FinType} is a packager~\cite[Sect.~7]{Mahboubi:2013} (see Appendix \ref{sec:invariant-concrete-modularity}) that takes a type \coqinline{T} and a \coqinline{Finite} mixin of \coqinline{T} as its arguments and construct a \coqinline{finType} instance from the given mixin and the canonical \coqinline{countType} instance for \coqinline{T}. 
However, if one omits its first argument \coqinline{T} with a placeholder as in the above, the packager may behave unpredictably as a unification hint.
In the above case, the placeholder was instantiated with \coqinline{extremal_group_countType} by type inference; as a result, it incorrectly overwrote the join of \coqinline{finType} and \coqinline{countType}.

\pagebreak

Our tools can also help finding inheritance bugs when extending the hierarchy of \MC, improve the development process by reducing the reviewing and maintenance burden, and allow developers and contributors to focus better on mathematical contents and other design issues.
For instance, Hivert~\cite{MathComp:PR406} added new structures of commutative algebras and redefined the field extension and splitting field structures to inherit from them.
In this extension process, he fixed some inheritance issues with help from us and our tools; at the same time, we made sure there is no inheritance bug without reviewing the whole boilerplate code of structures.
We ported the \textsf{order} sub-library of the \texttt{mathcomp-finmap} library~\cite{mathcomp-finmap-github} to \MC{}, redefined numeric domain structures~\cite[Chap.~4]{Cohen:phd}\cite[Sect.~3.1]{Cohen:2012} to inherit from ordered types, and factored out the notion of norms and absolute values as normed Abelian groups~\cite[Sect.~4.2]{forgetful-inference} with the help of our tools~\cite{MathComp:PR270,MathComp:PR453}.
This modification resulted in approximately 10,000 lines of changes; thus, reducing the reviewing burden was an even more critical issue.
This work is motivated by an improvement of the \Analysis{} library~\cite{analysis-github}\cite[Part II]{Rouhling:2019}, which extends the hierarchy of \MC{} with some algebraic and topological structures~\cite[Sect.~4]{forgetful-inference}\cite[Chap.~5]{Rouhling:2019} and is another application of our tools.

\section{Conclusion and related work}
\label{sec:conclusion}

This paper has provided a thorough analysis of the packed classes methodology, introduced two invariants that ensure the modularity of reasoning and the predictability of structure inference, and presented systematic ways to check those invariants.
We implemented our invariant checking mechanisms as a part of the \Coq{} system and a tool bundled with \MC.
With the help of these tools, many inheritance bugs in \MC{} have been found and fixed. The \MC{} development process has also been improved significantly.

\Coq{} had no coherence checking mechanism before our work. Sa\"{\i}bi~\cite[Sect.~7]{Saibi:1997} claimed that the coherence property ``is too restrictive in practice'' and ``it is better to replace conversion by Leibniz equality to compare path coercions because Leibniz equality is a bigger relation than conversion''.
However, most proof assistants based on dependent type theories including \Coq{} still rely heavily on conversion, particularly in their type checking/inference mechanisms.
Coherence should not be relaxed with Leibniz equality; otherwise, the type mismatch problems described in Sect.~\ref{sec:coherence} will occur.
With our coherence checking mechanism, users can still declare inconvertible multiple inheritance at their own risk and responsibility, because ambiguous paths messages are implemented as warnings rather than errors.
The \Lean{} system has an implicit coercion mechanism based on type class resolution, that allows users to define and use non-uniform implicit coercions; thus, coherence checking can be more difficult.
Actually, \Lean{} has no coherence checking mechanism; thus, users get more flexibility with this approach but need to be careful about being coherent.

There are three kinds of approaches to defining mathematical structures in dependent type theories: \emph{unbundled}, \emph{semi-bundled}, and \emph{bundled} approaches~\cite[Sect.~4.1.1]{mathlib:2020}. \pagebreak
The unbundled approach uses an interface that is parameterized by carriers and operators, and gathers axioms as its fields, e.g.,~\cite{Spitters:2011}; in contrast, the semi-bundled approach bundles operators together with axioms as in \coqinline{class_of} records, but still places carriers as parameters, e.g.,~\cite{mathlib:2020}.
The bundled approach uses an interface that bundles carriers together with operators and axioms, e.g., packed classes and telescopes~\cite[Sect.~2.3]{Mahboubi:2013}\cite{DeBruijn:1991,Pollack:2002,Geuvers:2002}.
The above difference between definitions of interfaces, in particular, whether carriers are bundled or not, leads to the use of different instance resolution and inference mechanisms: type classes~\cite{Sozeau:2008,Haftmann:2007} for the unbundled and semi-bundled approaches, and canonical structures or other unification hint mechanisms for the bundled approach.
Researchers have observed unpredictable behaviors~\cite{Hales:2018:blog} and efficiency issues~\cite[Sect.~4.3]{mathlib:2020}\cite[Sect.~11]{Spitters:2011} in inference with type classes; in contrast, structure inference with packed classes is predictable, and Theorem \ref{thm:join-AC} states this predictability more formally, except for concrete instance resolution.
The resolution of canonical structures is carried out by consulting a table of unification hints indexed by pairs of two head symbols and optionally with its recursive application and backtracking~\cite[Sect.~2.3]{Gonthier:2013b}.
The packed classes methodology is designed to use this recursive resolution not for structure inference~\cite[Sect.~2.3]{Garillot:2009} but only for parametric instances~\cite[Sect.~4]{Mahboubi:2013} such as lists and products, and not to use backtracking.
Thus, there is no efficiency issue in structure inference, except that nested class records and chains of their projections exponentially slow down the conversion which flat variant of packed classes~\cite[Sect.~4]{hierarchy-builder} can mitigate.
In the unbundled and semi-bundled approaches, a carrier may be associated with multiple classes; thus, inference of join and our work on structure inference (Sect.~\ref{sec:automated-structure-inference}, \ref{sec:formal-hierarchy}, and \ref{sec:validating-canonical-projections}) are problems specific to the bundled approach.\kp{Maybe say something about exponential blowups for unbundled as discussed here: \url{https://www.ralfj.de/blog/2019/05/15/typeclasses-exponential-blowup.html}}
A detailed comparison of type classes and packed classes has also been provided in~\cite{forgetful-inference}.
There are a few mechanisms to extend the unification engines of proof assistants other than canonical structures that can implement structure inference for packed classes: unification hints~\cite{Asperti:2009} and coercions pullback~\cite{Coen:2007}.
For any of those cases, our invariants are fundamental properties to implement packed classes and structure inference, but the invariant checking we propose has not been made yet at all.

Packed classes require the systematic use of records, implicit coercions, and canonical structures.
This leads us to automated generation of structures from their higher-level descriptions~\cite{hierarchy-builder}, which is work in progress.

\paragraph{Acknowledgements}

We appreciate the support from the STAMP (formerly \linebreak MARELLE) project-team, INRIA Sophia Antipolis.
In particular, we would like to thank Cyril Cohen and Enrico Tassi for their help in understanding packed classes and implementing our checking tools.
We are grateful to Reynald Affeldt, Yoichi Hirai, Yukiyoshi Kameyama, Enrico Tassi, and the anonymous reviewers for helpful comments on early drafts.
We are deeply grateful to Karl Palmskog for his careful proofreading and comments.
We would like to thank the organizers and participants of The Coq Workshop 2019 where we presented preliminary work of this paper.
In particular, an insightful question from Damien Pous was the starting point of our formal model of hierarchies presented in Sect.~\ref{sec:formal-hierarchy}.
This work was supported by JSPS Research Fellowships for Young Scientists and JSPS KAKENHI Grant Number 17J01683.

\renewcommand{\doi}[1]{\url{https://doi.org/#1}}
\bibliography{bibliography}


\newpage
\appendix

\section{\Coq{} references}
\label{sec:coq-references}

The following definitions express the basic properties of algebraic operations and constants used in Sect.~\ref{sec:packed-classes}.
\begin{coqcode}
Section ssrfun.
|>Local Set Implicit Arguments<|.
Variables S T R : Type.
\end{coqcode}
\begin{itemize}
 \item \coqinline{left_inverse e inv op <->}
       \coqinline{inv} is a left inverse of \coqinline{op} with respect to identity \coqinline{e}, i.e., \coqinline{op (inv x) x = e} for any \coqinline{x}.
 \item \coqinline{right_inverse e inv op <->}
       \coqinline{inv} is a right inverse of \coqinline{op} with respect to identity \coqinline{e}, i.e., \coqinline{op x (inv x) = e} for any \coqinline{x}.
\end{itemize}
\begin{coqcode}
Definition left_inverse e inv (op : S -> T -> R) := forall x, op (inv x) x = e.
Definition right_inverse e inv (op : S -> T -> R) := forall x, op x (inv x) = e.
\end{coqcode}
\begin{itemize}
 \item \coqinline{left_id e op <->}
       \coqinline{e} is a left identity for \coqinline{op}, i.e., \coqinline{op e x = x} for any \coqinline{x}.
 \item \coqinline{right_id e op <->}
       \coqinline{e} is a right identity for \coqinline{op}, i.e., \coqinline{op x e = x} for any \coqinline{x}.
\end{itemize}
\begin{coqcode}
Definition left_id e (op : S -> T -> T) := forall x, op e x = x.
Definition right_id e (op : S -> T -> S) := forall x, op x e = x.
\end{coqcode}
\begin{itemize}
 \item \coqinline{left_zero z op <->}
       \coqinline{z} is a left absorbing for \coqinline{op}, i.e., \coqinline{op z x = z} for any \coqinline{x}.
 \item \coqinline{right_zero z op <->}
       \coqinline{z} is a right absorbing for \coqinline{op}, i.e., \coqinline{op x z = z} for any \coqinline{x}.
\end{itemize}
\begin{coqcode}
Definition left_zero z (op : S -> T -> S) := forall x, op z x = z.
Definition right_zero z (op : S -> T -> T) := forall x, op x z = z.
\end{coqcode}
\begin{itemize}
 \item \coqinline{left_distributive op add <->}
       \coqinline{op} is left distributive over \coqinline{add},

       i.e., \coqinline{op (add x y) z = add (op x z) (op y z)} for any \coqinline{x}, \coqinline{y}, and \coqinline{z}.
 \item \coqinline{right_distributive op add <->}
       \coqinline{op} is right distributive over \coqinline{add},

       i.e., \coqinline{op x (add y z) = add (op x y) (op x z)} for any \coqinline{x}, \coqinline{y}, and \coqinline{z}.
\end{itemize}
\begin{coqcode}
Definition left_distributive (op : S -> T -> S) add :=
  forall x y z, op (add x y) z = add (op x z) (op y z).
Definition right_distributive (op : S -> T -> T) add :=
  forall x y z, op x (add y z) = add (op x y) (op x z).
\end{coqcode}
\begin{itemize}
 \item \coqinline{commutative op <->}
       \coqinline{op} is commutative,

       i.e., \coqinline{op x y = op y x} for any \coqinline{x} and \coqinline{y}.
 \item \coqinline{associative op <->}
       \coqinline{op} is associative,

       i.e., \coqinline{op x (op y z) = op (op x y) z} for any \coqinline{x}, \coqinline{y}, and \coqinline{z}.
\end{itemize}
\begin{coqcode}
Definition commutative (op : S -> S -> T) := forall x y, op x y = op y x.
Definition associative (op : S -> S -> S) :=
  forall x y z, op x (op y z) = op (op x y) z.

End ssrfun.
\end{coqcode}

\section{The invariant of packed classes concerning concrete instances and implicit coercions}
\label{sec:invariant-concrete-modularity}

In this appendix, we study the invariant of the packed classes methodology concerning concrete instances and implicit coercions, that ensures the modularity of reasoning with regard to concrete instances.
We also revisit an existing mechanism to construct concrete instances, packagers~\cite[Sect.~7]{Mahboubi:2013}. The packager mechanism is originally presented to make instance declarations easier, but also enforces this invariant to concrete instances.
We have not made any systematic checking mechanism for this invariant; thus, this appendix does not provide any original results. Nonetheless, practically we do not need to check this invariant thanks to the packager mechanism.

The packed classes methodology uses canonical structures for two purposes: structure inference (Sect.~\ref{sec:automated-structure-inference}) and canonical instances declaration for concrete types.
For example, one may use the \coqinline{Canonical} command to declare the canonical \coqinline{Monoid.type}, \coqinline{Semiring.type}, \coqinline{Group.type}, and \coqinline{Ring.type} instances for binary integers \coqinline{ZArith.BinInt.Z} of the \Coq{} standard library as follows.
\begin{coqcode}
Require Import ZArith.

Definition Z_monoid : Monoid.mixin_of Z :=
  Monoid.Mixin Z 0%Z Z.add Z.add_assoc Z.add_0_l Z.add_0_r.

Canonical Z_monoidType := Monoid.Pack Z (Monoid.Class Z Z_monoid).

Definition Z_semiring : Semiring.mixin_of Z_monoidType :=
  Semiring.Mixin Z_monoidType
                 1%Z Z.mul Z.add_comm Z.mul_assoc Z.mul_1_l Z.mul_1_r
                 Z.mul_add_distr_r Z.mul_add_distr_l Z.mul_0_l Z.mul_0_r.

Canonical Z_semiringType :=
  Semiring.Pack Z (Semiring.Class Z (Monoid.Class Z Z_monoid) Z_semiring).

Definition Z_group : Group.mixin_of Z_monoidType :=
  Group.Mixin Z_monoidType Z.opp Z.add_opp_diag_l Z.add_opp_diag_r.

Canonical Z_groupType :=
  Group.Pack Z (Group.Class Z (Monoid.Class Z Z_monoid) Z_group).

Canonical Z_ringType :=
  Ring.Pack Z (Ring.Class Z (Group.Class Z (Monoid.Class Z Z_monoid) Z_group)
                          Z_semiring).
\end{coqcode}
By declaring the first canonical instance \coqinline{Z_monoidType}, the unification algorithm of \Coq{} can solve a unification problem $\text{\coqinline{Monoid.sort}} ~ \evar{M} \unify \text{\coqinline{Z}}$ by instantiating $\evar{M} : \text{\coqinline{Monoid.type}}$ with \coqinline{Z_monoidType}.
Therefore, the type inference algorithm of \Coq{} can instantiate the unification variable $\evar{M} : \text{\coqinline{Monoid.type}}$ in a \Gallina{} term \coqinline{@add $\evar{M}$ 0
Other \coqinline{Canonical} declarations respectively work similarly as unification hints to infer canonical semiring, group, and ring instances for \coqinline{Z}.

The initial motivation to introduce packagers in~\cite[Sect.~7]{Mahboubi:2013} was not to repeat the same construction of class instances in canonical instance declarations.
For instance, all the above canonical instance declarations have subterm \coqinline{Monoid.Class Z Z_monoid} inside; moreover, \coqinline{Z_ringType} is consisting of the monoid, semiring, and group mixins of \coqinline{Z}, and thus can be automatically constructed from canonical semiring and group instances in principle.
These redundancies exist because the \coqinline{class_of} record must be the set of all mixin of superclasses (Sect.~\ref{sec:packed-classes}); in other words, inheritance in packed classes should always be by inclusion rather than construction/theorem.
Thanks to these redundancies, canonical concrete instances can always be defined to respect the following proposition, which is the invariant concerning implicit coercions and canonical concrete instances.

\begin{proposition}
 \label{invariant:concrete-definitional-equality}
 For any structures \coqinline{A} and \coqinline{B}, and a concrete type \coqinline{T : Type}, such that \coqinline{B} inherits from \coqinline{A} with an implicit coercion $\text{\coqinline{B.aType}} : \text{\coqinline{B.type}} \rightarrowtail \text{\coqinline{A.type}}$ and \coqinline{T} has the canonical instances \coqinline{T_aType} (\coqinline{:= A.Pack T ...}) and \coqinline{T_bType} (\coqinline{:= B.Pack T ...}), the following definitional equation holds:
 \[
  \text{\coqinline{B.aType T_bType}} \conv \text{\coqinline{T_aType}}.
 \]
\end{proposition}

By substituting our example of Sect.~\ref{sec:packed-classes} and their instances for \coqinline{Z} into Proposition \ref{invariant:concrete-definitional-equality}, the following definitional equations can be obtained:
\begin{align}
 \text{\coqinline{Semiring.monoidType Z_semiringType}}
   &\conv \text{\coqinline{Z_monoidType}}, \label{eq:Z-semiring-monoid} \\
 \text{\coqinline{Group.monoidType Z_groupType}}
   &\conv \text{\coqinline{Z_monoidType}}, \label{eq:Z-group-monoid} \\
 \text{\coqinline{Ring.monoidType Z_ringType}}
   &\conv \text{\coqinline{Z_monoidType}}, \label{eq:Z-ring-monoid} \\
 \text{\coqinline{Ring.semiringType Z_ringType}}
   &\conv \text{\coqinline{Z_semiringType}}, \label{eq:Z-ring-semiring} \\
 \text{\coqinline{Ring.groupType Z_ringType}}
   &\conv \text{\coqinline{Z_groupType}}. \label{eq:Z-ring-group}
\end{align}
If Proposition \ref{invariant:concrete-definitional-equality} does not hold, reasonings regarding concrete instances using overloaded operators and lemmas cannot be modular.
For example, we cannot reuse definitions and theorems on \coqinline{Z_monoidType} to reason about \coqinline{Z_semiringType} if equation (\ref{eq:Z-semiring-monoid}) does not hold.
We discuss about this problem in detail in~\cite[Sect.~3]{forgetful-inference}; also, the same issue and solution has recently been observed in~\cite[Sect.~3]{Buzzard:2020} which uses a bundled type classes based approach.

The packager mechanism provides a way to build a concrete instance of a structure from a concrete type and a mixin of that structure, but mixins of its superclasses are not required.
Since the monoid structure has no superclasses, the monoid packager can be defined straightforwardly as the following notation.
\begin{coqcode}
Notation MonoidType T m := (Monoid.Pack T (Monoid.Class T m)).
\end{coqcode}
The canonical monoid instance for \coqinline{Z} can be redefined by using this packager as follows.
\begin{coqcode}
Canonical Z_monoidType := MonoidType Z Z_monoid.
\end{coqcode}

The semiring packager takes a type and its semiring mixin, fetches the canonical monoid instance for the given type by using a phantom type, and build a semiring instance from them.
This packager can be defined as follows by using the infrastructure provided in~\cite[Sect.~7]{Mahboubi:2013}, except that binders in its \coqinline{[find ...]} notations are extended with explicit type annotation.
\begin{coqcode}[literate=*{=?=}{{$\thicksim$}}2]
Definition pack_semiring (T : Type) :=
  [find bT : Monoid.type      | T =?= Monoid.sort bT | "is not an Monoid.type" ]
  [find b : Monoid.class_of T | b =?= Monoid.class bT ]
  fun m : Semiring.mixin_of (Monoid.Pack T b) =>
  Semiring.Pack T (Semiring.Class T b m).

Notation SemiringType T m := (pack_semiring T _ id_phant _ id_phant m).
\end{coqcode}
This semiring packager notation \coqinline{SemiringType T m} can be read as: ``find \coqinline{bT} of type \coqinline{Monoid.type} such that \coqinline{T} unifies with \coqinline{Monoid.sort bT}, find \coqinline{b} of type \coqinline{Monoid.class_of T} such that \coqinline{b} unifies with \coqinline{Monoid.class bT}, and build a semiring instance \coqinline{Semiring.Pack T (Semiring.Class T b m)}''.
The first occurrence of the \coqinline{[find ...]} notation in \coqinline{pack_semiring} hides two binders
\begin{coqcode}
fun (bT : Monoid.type) (_ : phantom T -> phantom (Monoid.sort bT)) => ...
\end{coqcode}
where \coqinline{phantom} is defined as follows:
\begin{coqcode}
Variant phantom {T : Type} (p : T) : Prop := Phantom : phantom p.
\end{coqcode}
The \coqinline{SemiringType} notation instantiates the above binders with two arguments \coqinline{_ id_phant} where \coqinline{id_phant} is an identity function of type \coqinline{phantom t -> phantom t} (for some \coqinline{t}), so that \coqinline{T} unifies with \coqinline{Monoid.sort bT} to instantiates \coqinline{bT} with the canonical monoid instance for \coqinline{T}.
The second occurrence of the \coqinline{[find ...]} notation coerces \coqinline{Monoid.class bT} of type \coqinline{Monoid.class_of |*bT*|} to type \coqinline{Monoid.class_of |*T*|} by unification.

The canonical semiring instance for \coqinline{Z} can be redefined by using the semiring packager as follows.
\begin{coqcode}
Canonical Z_semiringType := SemiringType Z Z_semiring.
\end{coqcode}
Notably, if we declare the monoid instance \coqinline{Z_monoidType} as \coqinline{Definition} instead of \coqinline{Canonical}, type-checking the above declaration \coqinline{Z_semiringType} fails with the message
\begin{coqcode}
`Error: Z "is not an Monoid.type"
\end{coqcode}
so that the user can find a fact that the canonical monoid instance for \coqinline{Z} has not been declared yet.

As a side note, if one omits the type argument of a packager with a placeholder as follows, the packager may behave unpredictably as a unification hint as we mentioned in Sect.~\ref{sec:evaluation}.
\begin{coqcode}
Canonical bad_Z_semiringType := SemiringType _ Z_semiring.
\end{coqcode}
The placeholder in the above diclaration is instantiated with \coqinline{Monoid.type Z_monoidType} by the first occurrence of \coqinline{[find ...]} notation in \coqinline{pack_semiring}; thus, a unification hint that overwrites the join of semiring and monoids with the concrete type \coqinline{Z} will be synthesized, and the assertion for that join
\begin{coqcode}
check_join Semiring.type Monoid.type Semiring.type.
\end{coqcode}
fails with the following message:
\begin{code}
The join of Semiring.type and Monoid.type is a concrete type Z but is expected
to be Semiring.type.
\end{code}

The group packager and the canonical group instance for \coqinline{Z} can similarly be (re)defined as follows.
\begin{coqcode}[literate=*{=?=}{{$\thicksim$}}2]
Definition pack_group (T : Type) :=
  [find bT : Monoid.type      | T =?= Monoid.sort bT | "is not an Monoid.type" ]
  [find b : Monoid.class_of T | b =?= Monoid.class bT ]
  fun m : Group.mixin_of (Monoid.Pack T b) =>
  Group.Pack T (Group.Class T b m).

Notation GroupType T m := (pack_group T _ id_phant _ id_phant m).

Canonical Z_groupType := GroupType Z Z_group.
\end{coqcode}

Since the ring structure inherits from groups and semirings and has no further axiom, the ring packager can be defined as the following notation that only takes a type; internally, it finds the canonical group and semiring instances for the given type and construct a ring instance from them.
\begin{coqcode}[literate=*{=?=}{{$\thicksim$}}2]
Definition pack_ring (T : Type) :=
  [find gT : Group.type      | T =?= Group.sort gT | "is not an Group.type" ]
  [find g : Group.class_of T | g =?= Group.class gT ]
  [find sT : Semiring.type   | T =?= Semiring.sort sT
                             | "is not an Semiring.type" ]
  [find s : Semiring.mixin_of (Monoid.Pack T (Group.base T g))
                             | s =?= Semiring.mixin _ (Semiring.class sT) ]
  Ring.Pack T (Ring.Class T g s).

Notation RingType T :=
  (pack_ring T _ id_phant _ id_phant _ id_phant _ id_phant).

Canonical Z_ringType := RingType Z.
\end{coqcode}

In the above definitions, each packager notation of a structure takes canonical instances of its (direct) superclasses for a given type, and constructs an instance of the structure using mixins and classes included in that canonical instances; thus, they enforce Proposition \ref{invariant:concrete-definitional-equality} on canonical instances, which can be confirmed by the following \Ltac{} script.
\begin{coqcode}
unify (Semiring.monoidType Z_semiringType) Z_monoidType.
unify (Group.monoidType Z_groupType)       Z_monoidType.
unify (Ring.monoidType Z_ringType)         Z_monoidType.
unify (Ring.semiringType Z_ringType)       Z_semiringType.
unify (Ring.groupType Z_ringType)          Z_groupType.
\end{coqcode}

In general, if packagers are defined and used properly, Proposition \ref{invariant:concrete-definitional-equality} should hold.
Users can still bypass the packager mechanism and construct canonical instances that do not respect the invariant, but practically packagers work to enforce the invariant.

\section{Omitted lemmas and proofs for Sect.~\ref{sec:formal-hierarchy}}
\label{sec:formal-hierarchy-proofs}

This appendix shows informal statements and proofs of lemmas that prove Theorem \ref{thm:join-AC} and the complete formalization of Sect.~\ref{sec:formal-hierarchy} in \Coq. For each part, informal ones come before formal ones.

\begin{coqcode}
From mathcomp Require Import ssreflect ssrfun ssrbool eqtype ssrnat seq choice.
From mathcomp Require Import fintype order finset.

|>Set Implicit Arguments.<|
Unset Strict Implicit.
Unset Printing Implicit Defensive.

Import Order.Syntax Order.LTheory.
\end{coqcode}

\begin{lemma}
 \label{lem:finltgt-wf}
 The strict inheritance relation $\leadsto^+$ and its inverse relation $\leadsfrom^+$ are well-founded relations.
\end{lemma}

\begin{proof}
 Strict partial orders on a finite set are well-founded. \qedhere
\end{proof}

\begin{coqcode}
Lemma finlt_wf disp (T : finPOrderType disp) : well_founded (<%O : rel T)%O.
Proof.
move=> x.
have: #|< x|%O <= #|T| by exact: max_card.
elim: #|T| x => [|n ihn] x hx; constructor => y hy.
- have: y \in enum (< x)%O by rewrite mem_enum.
  by move: hx; rewrite cardE /=; case: (enum _).
- apply: ihn; rewrite -ltnS; apply/(leq_trans _ hx)/proper_card.
  apply/properP; split.
  + by apply/subsetP => ?; move/lt_trans; apply.
  + by exists y => //; rewrite inE ltxx.
Qed.

Lemma fingt_wf disp (T : finPOrderType disp) : well_founded (>%O : rel T)%O.
Proof. exact: (@finlt_wf _ [finPOrderType of Order.converse T]). Qed.
\end{coqcode}

Let us assume \coqinline{S} is the type of structures whose ordering is the inheritance relation: \coqinline{a <= b} and \coqinline{a < b} mean \coqinline{b} non-strictly and strictly inherits from \coqinline{a} respectively.

\begin{coqcode}
Section join_structures_meta_properties.
Variables (disp : unit) (S : finPOrderType disp).
\end{coqcode}

\subsubsection*{Definition \ref{def:common-subclasses}} is as follows.

\begin{coqcode}
(* Weak and strict subclasses.                                                *)
Definition w_subclasses (b : S) : {set S} := [set a | b <= a]%O.
Definition s_subclasses (b : S) : {set S} := [set a | b < a]%O.

(* Minimal common subclasses.                                                 *)
Definition mcs (a b : S) : {set S} :=
  let subs := w_subclasses a :&: w_subclasses b in
  subs :\: \bigcup_(c in subs) s_subclasses c.
\end{coqcode}

\subsubsection*{Definition \ref{def:well-formed-hierarchy}} is as follows.

\begin{coqcode}
(* For any structure `a` and `b`, their join must be unique.                  *)
Hypothesis (unique_join : forall a b : S, #|mcs a b| <= 1).
\end{coqcode}

\subsubsection*{Definition \ref{def:extended-hierarchy}} is as follows.

\begin{coqcode}
(* The type of structures extended with the top (undefined) structure.        *)
Definition extS := option S.

(* Ordering of `S` can be naturally extended to `extS` as follows.            *)
Definition le_extS (a b : extS) : bool :=
  match b, a with
  | None, _ => true
  | Some _, None => false
  | Some b, Some a => a <= b
  end%O.

Definition lt_extS (a b : extS) : bool :=
  match a, b with
  | None, None => false
  | Some _, None => true
  | None, Some _ => false
  | Some a, Some b => a < b
  end%O.
\end{coqcode}

\begin{lemma}
 An extended hierarchy $\bar{\mathcal{H}}$ is (non-strictly) partially ordered by its inheritance relation $\bar{\leadsto}^*$.
\end{lemma}

\begin{proof}
 The reflexivity, antisymmetry and transitivity of $\bar{\leadsto}^*$ should be checked here.
 The proofs are by case analysis. \qedhere
\end{proof}

\begin{coqcode}
Lemma lt_extS_def a b : lt_extS a b = (b != a) && le_extS a b.
Proof. by case: a b => [a|][b|] //=; rewrite lt_def. Qed.

Lemma le_extS_refl : reflexive le_extS. Proof. by case=> /=. Qed.

Lemma le_extS_anti : antisymmetric le_extS.
Proof. by case=> [a|][b|] // /le_anti ->. Qed.

Lemma le_extS_trans : transitive le_extS.
Proof. by case=> [?|][?|][?|] //; exact: le_trans. Qed.

Definition extS_porderMixin :=
  LePOrderMixin lt_extS_def le_extS_refl le_extS_anti le_extS_trans.
Canonical extS_porderType := POrderType disp extS extS_porderMixin.

Lemma le_extS_SomeE a b : (Some a <= Some b :> extS)%O = (a <= b)%O.
Proof. by []. Qed.

Lemma lt_extS_SomeE a b : (Some a < Some b :> extS)%O = (a < b)%O.
Proof. by []. Qed.
\end{coqcode}

\subsubsection*{Definition \ref{def:unique-join}} is as follows.

\begin{coqcode}
Definition join (a b : extS) : extS :=
  match a, b with
    | Some a, Some b =>
      if enum (mcs a b) is c :: _ then Some c else None
    | _, _ => None
  end.
\end{coqcode}

\begin{lemma}
 \label{lem:joinN}
 For an extended hierarchy, $\top$ is the absorbing element.
\end{lemma}

\begin{proof}
 By definition. \qedhere
\end{proof}

\begin{coqcode}
Lemma joinNx : left_zero None join. Proof. done. Qed.
Lemma joinxN : right_zero None join. Proof. by case. Qed.
\end{coqcode}

\begin{lemma}
 \label{lem:joinxx}
 For an extended hierarchy, the join is idempotent:
 \[
  \forall A \in \bar{\mathcal{H}}, \mathrm{join}(A, A) = A.
 \]
\end{lemma}

\begin{proof}
 If $A = \top$, $\mathrm{join}(A, A) = A$ is true by Definition \ref{def:unique-join}; otherwise, it is sufficient to show $\mathrm{mcs}(A, A) = \{A\}$.
 Let us prove they are extensionally equal.
 \begin{align*}
  & B \in \mathrm{mcs}(A, A) \\
  \Leftrightarrow{}
  & A \leadsto^* B \land \neg (\exists A' \in \mathcal{H}, A \leadsto^* A' \land A' \leadsto^+ B)
  & \text{(Definition \ref{def:common-subclasses})}
  \intertext{
  If $A \leadsto^* B$ does not hold, the above formula is false; also, the RHS $B \in \{A\}$ ($\Leftrightarrow A \leadsto^* B \land B \leadsto^* A$) is false. Thus, let us consider the case $A \leadsto^* B$ holds.
  }
  \Leftrightarrow{}
  & \neg (\exists A' \in \mathcal{H}, A \leadsto^* A' \land A' \leadsto^+ B) \\
  \Leftrightarrow{}
  & \neg (A \leadsto^+ B) & \text{(Transitivity)} \\
  \Leftrightarrow{}
  & \neg (A \leadsto^* B \land A \neq B) \\
  \Leftrightarrow{}
  & \neg (A \neq B) & \text{(Assumption)} \\
  \Leftrightarrow{}
  & B \in \{A\}
 \end{align*}
 The above transformation shows $\mathrm{mcs}(A, A) = \{A\}$; thus, the join is idempotent. \qedhere
\end{proof}

\begin{coqcode}
Lemma joinxx : idempotent join.
Proof.
case=> [a|] //=; rewrite /mcs setIid; set A := _ :\: _.
suff ->: A = [set a] by rewrite enum_set1.
rewrite {}/A; apply/setP => b.
rewrite !inE eq_le andbC [RHS]andbC; case: (boolP (a <= b)%O) => //= Hab.
apply/bigcupP; case: ifP => [Hba [c]| Hba].
- by rewrite !inE => Hac; move/(le_lt_trans (le_trans Hba Hac)); rewrite ltxx.
- by exists a; rewrite !inE //; rewrite lt_def Hab andbT eq_le Hba.
Qed.
\end{coqcode}

\begin{lemma}
 \label{lem:joinC}
 For an extended hierarchy, the join is commutative:
 \[
  \forall A \, B \in \bar{\mathcal{H}}, \mathrm{join}(A, B) = \mathrm{join}(B, A).
 \]
\end{lemma}

\begin{proof}
 Since the join is defined symmetrically with respect to the parameters in Definitions \ref{def:common-subclasses} and \ref{def:unique-join}, it is commutative. \qedhere
\end{proof}

\begin{coqcode}
Lemma joinC : commutative join.
Proof. by move=> [a|][b|] //=; rewrite /mcs setIC. Qed.
\end{coqcode}

\begin{lemma}
 \label{lem:join-inherit-lr}
 For any structures $A$ and $B$ in an extended hierarchy $\bar{\mathcal{H}}$, $\mathrm{join}(A, B)$ non-strictly inherits from both $A$ and $B$:
 \begin{align*}
  A \mathrel{\bar{\leadsto}}^* \mathrm{join}(A, B), &&
  B \mathrel{\bar{\leadsto}}^* \mathrm{join}(A, B).
 \end{align*}
\end{lemma}

\begin{proof}
 If $A = \top$, $B = \top$, or $\mathrm{mcs}(A, B)$ is not a singleton set, $\mathrm{join}(A, B) = \top$ holds; thus, it inherits from both $A$ and $B$.
 Otherwise, $\mathrm{mcs}(A, B)$ is a singleton set $\{C\}$ and $\mathrm{join}(A, B) = C$ holds. $C$ inherits from both $A$ and $B$ by Definition \ref{def:common-subclasses}. \qedhere
\end{proof}

\begin{coqcode}
Lemma join_inherit_l a b : (a <= join a b)%O.
Proof.
case: a b=> [a|][b|] //=.
suff: forall c, c \in enum (mcs a b) -> (a <= c)%O
  by case: (enum _) => // c ? /(_ c); rewrite inE eqxx /=; exact.
by move=> c; rewrite mem_enum; case/setDP => /setIP [/=]; rewrite inE.
Qed.

Lemma join_inherit_r a b : (b <= join a b)%O.
Proof. by rewrite joinC join_inherit_l. Qed.
\end{coqcode}

\begin{lemma}
 \label{lem:le-joinE}
 For any structures $A$, $B$, and $C$ in an extended well-formed hierarchy $\bar{\mathcal{H}}$, $C$ non-strictly inherits from $\mathrm{join}(A, B)$ if and only if $C$ non-strictly inherits from both $A$ and $B$:
 \[
  \mathrm{join}(A, B) \mathrel{\bar{\leadsto}}^* C \iff A \mathrel{\bar{\leadsto}}^* C \land B \mathrel{\bar{\leadsto}}^* C.
 \]
\end{lemma}

\begin{proof}
 If $A$, $B$, or $C$ is $\top$, this proposition is trivial by Definitions \ref{def:extended-hierarchy} and \ref{def:unique-join}; also, the direct implication is trivial by Lemma \ref{lem:join-inherit-lr} and the transitivity of $\mathrel{\bar{\leadsto}}^*$.
 Thus it is sufficient to prove the converse implication for $A, B, C \in \mathcal{H}$.
 Let us name $\mathcal{C}$ the common subclasses of $A$ and $B$, which is $\{C \in \mathcal{H} \mid A \leadsto^* C \land B \leadsto ^* C\}$, and prove the following proposition by well-founded induction on $C \in \mathcal{H}$:
 \[
  C \in \mathcal{C} \Rightarrow \mathrm{join}(A, B) \mathrel{\bar{\leadsto}}^* C.
 \]
 If $\mathrm{mcs}(A, B)$ is empty, $\mathcal{C}$ is a subset of $\{C \in \mathcal{H} \mid \exists C' \in \mathcal{C}, C' \leadsto^+ C\}$ by Definition \ref{def:common-subclasses}; thus, $C' \leadsto^+ C$ holds for some $C' \in \mathcal{C}$ by the assumption $C \in \mathcal{C}$.
 By the induction hypothesis on $C'$, $\mathrm{join}(A, B) \mathrel{\bar{\leadsto}}^* C' ~ (\leadsto^+ C)$ holds.

 If $\mathrm{mcs}(A, B)$ is not empty, it is a singleton set $\{\mathrm{join}(A, B)\}$ because of the uniqueness of join; then, the following equation holds by Definition \ref{def:common-subclasses}:
 \[
  \mathcal{C} = \{\mathrm{join}(A, B)\} \cup \{C \in \mathcal{H} \mid \exists C' \in \mathcal{C}, C' \leadsto^+ C\}.
 \]
 By the assumption, $C$ inhabits $\{\mathrm{join}(A, B)\}$ or $\{C \in \mathcal{H} \mid \exists C' \in \mathcal{C}, C' \leadsto^+ C\}$. If $C = \mathrm{join}(A, B)$, the conclusion is true by reflexivity; otherwise, $C' \leadsto^+ C$ holds for some $C' \in \mathcal{C}$. By the induction hypothesis on $C'$, $\mathrm{join}(A, B) \mathrel{\bar{\leadsto}}^* C' ~ (\leadsto^+ C)$ holds. \qedhere
\end{proof}

\begin{coqcode}
Lemma le_joinE a b c : (join a b <= c)%O = (a <= c)%O && (b <= c)%O.
Proof.
apply/idP/idP => [le_ab_c|/andP []].
  rewrite (le_trans (join_inherit_l _ _) le_ab_c).
  exact: (le_trans (join_inherit_r _ _) le_ab_c).
case: a b c => [a|][b|][c|] //; rewrite !le_extS_SomeE => le_ac le_bc.
set sub_ab := w_subclasses a :&: w_subclasses b.
have {le_ac le_bc}: c \in sub_ab by rewrite !inE; apply/andP; split.
elim/Acc_ind: c / (finlt_wf c) => c _ ihc hc.
move: (unique_join a b); rewrite leq_eqVlt ltnS leqn0 cards_eq0 setD_eq0.
case/orP=> [/cards1P [join_ab joinsE] | hsub]; last first.
  case/(subsetP hsub)/bigcupP: (hc) => y hy; rewrite inE => lt_yx.
  exact/(le_trans (ihc _ lt_yx hy))/ltW/lt_yx.
move: ihc hc; rewrite /join joinsE enum_set1 le_extS_SomeE => ihc.
suff ->: sub_ab = join_ab |: \bigcup_(c in sub_ab) s_subclasses c.
  rewrite 2!inE => /predU1P [-> // | /bigcupP [c' hc']].
  by rewrite inE => lt_c'c; apply/(le_trans _ (ltW lt_c'c))/ihc.
rewrite -joinsE /mcs setDE setUIl (setUC (~: _)) setUCr setIT.
apply/esym/setUidPl/subsetP => x /bigcupP [y]; rewrite !inE.
by case/andP=> [? ?] lt_yx; rewrite !(le_trans _ (ltW lt_yx)).
Qed.
\end{coqcode}

\begin{lemma}
 \label{lem:joinA}
 For an extended well-formed hierarchy, the join is associative:
 \[
  \forall A \, B \, C \in \bar{\mathcal{H}}, \mathrm{join}(\mathrm{join}(A, B), C) = \mathrm{join}(A, \mathrm{join}(B, C)).
 \]
\end{lemma}

\begin{proof}
 The above equation has the following equivalent formula:
 \begin{align*}
  & \mathrm{join}(\mathrm{join}(A, B), C) = \mathrm{join}(A, \mathrm{join}(B, C)) \\
  \Leftrightarrow{}
  & \mathrm{join}(\mathrm{join}(A, B), C) \mathrel{\bar{\leadsto}}^* \mathrm{join}(A, \mathrm{join}(B, C)) \land {} \\
  & \mathrm{join}(A, \mathrm{join}(B, C)) \mathrel{\bar{\leadsto}}^* \mathrm{join}(\mathrm{join}(A, B), C) \\
  \intertext{\raggedleft(antisymmetry and reflexivity of $\mathrel{\bar{\leadsto}}^*$)}
  \Leftrightarrow{}
  & A \mathrel{\bar{\leadsto}}^* \mathrm{join}(A, \mathrm{join}(B, C)) \land
    B \mathrel{\bar{\leadsto}}^* \mathrm{join}(A, \mathrm{join}(B, C)) \land {} \\
  & C \mathrel{\bar{\leadsto}}^* \mathrm{join}(A, \mathrm{join}(B, C)) \land
    A \mathrel{\bar{\leadsto}}^* \mathrm{join}(\mathrm{join}(A, B), C) \land {} \\
  & B \mathrel{\bar{\leadsto}}^* \mathrm{join}(\mathrm{join}(A, B), C) \land
    C \mathrel{\bar{\leadsto}}^* \mathrm{join}(\mathrm{join}(A, B), C)
 \end{align*}
 \begin{flushright}
  (Lemma \ref{lem:le-joinE}).
 \end{flushright}

 Each clause of the above formula is a corollary of Lemma \ref{lem:join-inherit-lr}. \qedhere
\end{proof}

\begin{coqcode}
Lemma joinA : associative join.
Proof.
move=> a b c; apply/eqP; rewrite (eq_le _ (join _ _));
  apply/andP; split; rewrite !le_joinE -?andbA; apply/and3P; split.
- exact/(le_trans (join_inherit_l _ _))/join_inherit_l.
- exact/(le_trans (join_inherit_r _ _))/join_inherit_l.
- exact: join_inherit_r.
- exact: join_inherit_l.
- exact/(le_trans (join_inherit_l _ _))/join_inherit_r.
- exact/(le_trans (join_inherit_r _ _))/join_inherit_r.
Qed.
\end{coqcode}

Theorem \ref{thm:join-AC} consists of Lemmas \ref{lem:joinxx}, \ref{lem:joinC}, and \ref{lem:joinA}; thus holds.

\begin{coqcode}
End join_structures_meta_properties.
\end{coqcode}

\section{Evolution of the hierarchy in \MC}
\label{sec:evolution-mathcomp}

In this appendix, we report in detail recent efforts on fixing and extending the hierarchy of structures in the \MC{} library.
In those improvements, our tools help us and other contributors to \MC{} to find inheritance bugs, moreover, improve the development process by reducing the reviewing and maintenance burden, and allow developers and contributors to focus on mathematical contents and other design issues.

\subsection{\MC{} 1.7.0}

\begin{figure*}[t]
 \centering
 \includegraphics[width=\textwidth,pagebox=cropbox]{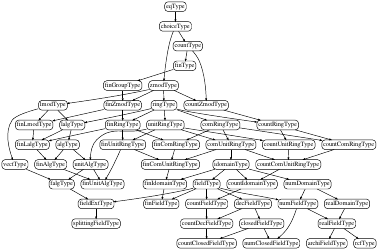}
 \caption{The hierarchy of structures in the \MC{} library 1.7.0}
 \label{fig:hierarchy-1_7_0}
\end{figure*}

We start with the hierarchy of \MC{} 1.7.0 which is depicted in Fig.~\ref{fig:hierarchy-1_7_0}, lacks 7 inheritance edges from the \coqinline{CountRing} structures to the \coqinline{FinRing} structures, and contains some other inheritance bugs.
Those missing edges do not mean just missing inference rules, but also breaks the well-formedness invariant.
For instance, \coqinline{countType} and \coqinline{zmodType} have two ambiguous joins \coqinline{countZmodType} and \coqinline{finZmodType} due to the missing edge from the former one to the latter one. The same issues can be seen in the \coqinline{ringType}, \coqinline{comRingType}, \coqinline{unitRingType}, \coqinline{comUnitRingType}, \coqinline{idomainType}, and \coqinline{fieldType} structures, and their corresponding countable and finite structures.
There are also the following 8 missing canonical projections:
\begin{itemize}
 \item the join of \coqinline{countClosedFieldType} and \coqinline{countIdomainType} which should be \coqinline{countClosedFieldType},
 \item the join of \coqinline{countType} and \coqinline{lmodType} which should be \coqinline{finLmodType},
 \item the join of \coqinline{countType} and \coqinline{lalgType} which should be \coqinline{finLalgType},
 \item the join of \coqinline{countType} and \coqinline{algType} which should be \coqinline{finAlgType},
 \item the join of \coqinline{countType} and \coqinline{unitAlgType} which should be \coqinline{finUnitAlgType},
 \item the join of \coqinline{finUnitRingType} and \coqinline{unitAlgType} which should be \coqinline{finUnitAlgType},
 \item the join of \coqinline{fieldType} and \coqinline{numDomainType} which should be \coqinline{numFieldType}, and
 \item the join of \coqinline{fieldType} and \coqinline{realDomainType} which should be \coqinline{realFieldType}.
\end{itemize}
The join of \coqinline{finType} and \coqinline{countType} should be \coqinline{finType}, but has been overwritten by the canonical \coqinline{finType} instance for \coqinline{extremal_group} as explained in Sect.~\ref{sec:evaluation}.

\subsection{\MC{} 1.8.0 and 1.9.0}

\begin{figure*}[t]
 \centering
 \includegraphics[width=\textwidth,pagebox=cropbox]{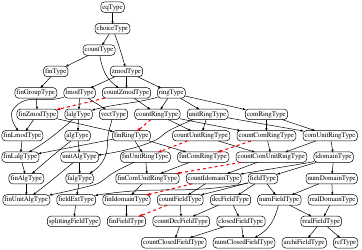}
 \caption{The hierarchy of structures in the \MC{} library 1.8.0 and 1.9.0.
 The bold dashed edges highlight new inheritance edges that are missing in 1.7.0.}
 \label{fig:hierarchy-1_8_0}
\end{figure*}

All the above inheritance issues in \MC{} 1.7.0 have been addressed in version 1.8.0 whose hierarchy of structures is depicted in Fig.~\ref{fig:hierarchy-1_8_0}. In this diagram, the bold dashed edges highlight new inheritance edges that are missing in 1.7.0.
Any implementation of structures and their inheritance has not been changed in \MC{} 1.9.0 from version 1.8.0.
Cohen~\cite{MathComp:PR209} made the initial attempt to add missing inheritances from \coqinline{CountRing} to \coqinline{FinRing} structures without any tooling work; however, this attempt was incomplete.
Then we made the second attempt to fix them with prior tooling work of Sect.~\ref{sec:validating-canonical-projections} by Cohen~\cite{MathComp:PR291}, which check only the existence of canonical projections but do not check which structure is inferred.
In this change, we added all the missing canonical projections thanks to the tool, but did not fix the canonical \coqinline{finType} instance for \coqinline{extremal_group} and also introduced a few misplaced canonical projections, because of the incompleteness of the checking.
Finally we made the initial complete implementation of the well-formedness checking~\cite{MathComp:PR318} as in Sect.~\ref{sec:validating-canonical-projections}, and fixed all the inheritance issues we discovered.

\subsection{\MC{} 1.10.0}

\begin{figure*}[t]
 \centering
 \includegraphics[width=\textwidth,pagebox=cropbox]{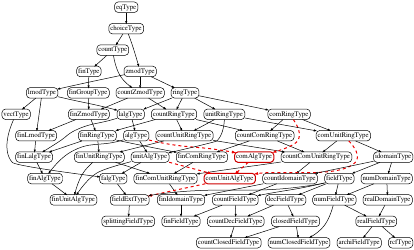}
 \caption{The hierarchy of structures in the \MC{} library 1.10.0.
 This diagram is identical with Fig.~\ref{fig:hierarchy} except that new structures and inheritance edges are highlighted by bold frames and bold dashed edges respectively.}
 \label{fig:hierarchy-1_10_0}
\end{figure*}

In \MC{} 1.10.0, Hivert~\cite{MathComp:PR406} added new structures \coqinline{comAlgType} and \coqinline{comUnitAlgType}, which are interfaces for non-unitary and unitary commutative algebras respectively.
The hierarchy of this version is depicted in Fig.~\ref{fig:hierarchy-1_10_0}, where the new structures and inheritance edges added are highlighted by bold frames and bold dashed edges respectively.
This extension requires to change the field extension (\coqinline{fieldExtType}) and splitting field (\coqinline{splittingFieldType}) structures to inherit from those commutative algebra structures, which involves some additions and deletions of joins; for instance, each join of \coqinline{lmodType}, \coqinline{lalgType}, \coqinline{algType}, \coqinline{unitAlgType}, and \coqinline{comRingType}, \coqinline{comUnitRingType} is \coqinline{fieldExtType} in 1.9.0, but should be generalized to one of the new commutative algebra structures.
In this extension process, they fixed some inheritance issues with help from us and our well-formedness checking tool; at the same time, we made sure there is no inheritance bug without reviewing the whole boilerplate code of implicit coercions and canonical projections.

\subsection{The current \texttt{master} branch of \MC{}}

In the current \texttt{master} branch of \MC{} which supposed to be released as a part of version 1.11.0, we ported the \textsf{order} sub-library of the \texttt{mathcomp-finmap} library~\cite{mathcomp-finmap-github} to \MC{}~\cite{MathComp:PR270,MathComp:PR453}.
The hierarchy of this version is depicted in Fig.~\ref{fig:hierarchy-master}, where the new structures and inheritance edges added are highlighted by bold frames and bold dashed edges respectively.
The \textsf{order} sub-library provides several structures of ordered types, including partially ordered types \coqinline{porderType}, (non-distributive) lattices \coqinline{latticeType}, distributive lattices \coqinline{distrLatticeType}, and totally ordered types \coqinline{orderType}; thus, this extension requires to change the numeric domain and numeric field structures~\cite[Chap.~4]{Cohen:phd}\cite[Sect.~3.1]{Cohen:2012} (\coqinline{numDomainType} and its subclasses) to inherit from those ordered type structures.
We also factored out normed Abelian groups from numeric domains as a new structure \coqinline{normedZmodType}, to unify several notions of norms in the \Analysis{} library~\cite[Sect.~4.2]{forgetful-inference}.
This refactoring work became approximately 10,000 lines of changes; thus, reducing the reviewing burden was an even more critical issue.

\begin{figure*}[t]
 \centering
 \includegraphics[width=\textwidth,pagebox=cropbox]{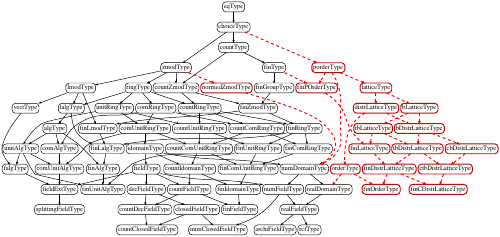}
 \caption{The hierarchy of structures in the \MC{} library (the \texttt{master} branch of the GitHub repository on Feb.~1, 2020, the commit hash is: \texttt{7d04173b52cf02717b8f8e8c13bb7c3521de7e89}).
 New structures and inheritance edges are highlighted by bold frames and bold dashed edges respectively.}
 \label{fig:hierarchy-master}
\end{figure*}

\end{document}